\documentclass[11pt]{article}
	
\usepackage[letterpaper,margin=0.90in]{geometry}
\usepackage{amsmath, amssymb, amsthm, amsfonts,bm}
\usepackage[inline]{enumitem}

\usepackage{cite}
\usepackage{framed}
\usepackage[framemethod=tikz]{mdframed}
\usepackage{appendix}
\usepackage{graphicx}
\usepackage{color}
\usepackage{wrapfig}
\usepackage{algorithm, subcaption}
\usepackage[noend]{algpseudocode}
\usepackage[textsize=tiny]{todonotes}

\usepackage{enumitem}
\setitemize{noitemsep,topsep=3pt,parsep=3pt,partopsep=3pt}

\usepackage{xspace}
\definecolor{darkgreen}{rgb}{0,0.5,0}
\definecolor{darkblue}{rgb}{0,0,0.8}
\usepackage{cleveref}

\theoremstyle{theorem}
\newtheorem{theorem}{Theorem}[section]
\theoremstyle{lemma}
\newtheorem{lemma}[theorem]{Lemma}
\theoremstyle{corollary}
\newtheorem{corollary}[theorem]{Corollary}
\theoremstyle{claim}
\newtheorem{claim}[theorem]{Claim}

\theoremstyle{definition}
\newtheorem{definition}{Definition}[section]
\theoremstyle{remark}

\newcommand{\ignore}[1]{}

\algnewcommand\algorithmicswitch{\textbf{switch}}
\algnewcommand\algorithmiccase{\textbf{case}}

\algdef{SE}[SWITCH]{Switch}{EndSwitch}[1]{\algorithmicswitch\ #1\ \algorithmicdo}{\algorithmicend\ \algorithmicswitch}%
\algdef{SE}[CASE]{Case}{EndCase}[1]{\algorithmiccase\ #1}{\algorithmicend\ \algorithmiccase}%
\algtext*{EndSwitch}%
\algtext*{EndCase}%
\newcommand{\eps}{\varepsilon}

\newcommand{\poly}{\operatorname{\text{{\rm poly}}}}
\newcommand{\polylog}{\operatorname{\text{{\rm polylog}}}}

\newcommand{\dist}{\operatorname{dist}}
\newcommand{\outdeg}{\operatorname{outdeg}}
\newcommand{\indeg}{\operatorname{indeg}}

\DeclareMathOperator{\E}{\mathbb{E}}

\renewcommand{\paragraph}[1]{\vspace{0.15cm}\noindent {\bf #1}:}

\newcommand{\FullOrShort}{short}

\ifthenelse{\equal{\FullOrShort}{full}}{
	
  \newcommand{\fullOnly}[1]{#1}
  \newcommand{\shortOnly}[1]{}

  }{

    \newcommand{\shortOnly}[1]{#1}
		\newcommand{\fullOnly}[1]{}
    
  }

\begin{document}

\date{}

\title{Distributed Degree Splitting, Edge Coloring, and Orientations} 

\author{
 Mohsen Ghaffari\\
  \small MIT \\
  \small ghaffari@mit.edu
\and
Hsin-Hao Su\\
	\small MIT \\
	\small hsinhao@mit.edu 
}

\maketitle

\begin{abstract} 
We study a family of closely-related distributed graph problems, which we call \emph{degree splitting}, where roughly speaking the objective is to partition (or orient) the edges such that each node's degree is split almost uniformly. Our findings lead to answers for a number of problems, a sampling of which includes:

\medskip
\begin{itemize} 

\item We present a $\poly\log n$ round deterministic algorithm for $(2\Delta-1)\cdot (1+o(1))$-\emph{edge-coloring}, where $\Delta$ denotes the maximum degree. Modulo the $1+o(1)$ factor, this settles one of the long-standing open problems of the area from the 1990's (see e.g. Panconesi and Srinivasan [PODC'92]). Indeed, a weaker requirement of $(2\Delta-1)\cdot \poly\log \Delta$-edge-coloring in $\poly\log n$ rounds was asked for in the 4th open question in the \emph{Distributed Graph Coloring} book by Barenboim and Elkin. 
\medskip

\item We show that \emph{sinkless orientation}---i.e., orienting edges such that each node has at least one outgoing edge---on $\Delta$-regular graphs can be solved in $O(\log_{\Delta} \log n)$ rounds randomized and in $O(\log_{\Delta} n)$ rounds deterministically. These prove the corresponding lower bounds by Brandt et al. [STOC'16] and Chang, Kopelowitz, and Pettie [FOCS'16] to be tight. Moreover, these show that sinkless orientation exhibits an exponential separation between its randomized and deterministic complexities, akin to the results of Chang et al. for $\Delta$-coloring $\Delta$-regular trees.

\medskip 
\item We present a randomized $O(\log^4 n)$ round algorithm for orienting $a$-arboricity graphs with maximum out-degree $a(1+\eps)$. This can be also turned into a decomposition into $a (1+\eps)$ forests when $a=\Omega(\log n)$ and into $a (1+\eps)$ pseduo-forests when $a=o(\log n)$. Obtaining an efficient distributed decomposition into less than $2a$ forests was stated as the 10th open problem in the book by Barenboim and Elkin. 

\end{itemize}   
\end{abstract}
%\listoftodos
\setcounter{page}{0}
\thispagestyle{empty}
%\newpage
%\tableofcontents
\newpage

\section{Introduction \& Related Work}
\vspace{-7pt}
Graph symmetry breaking problems form one of the central subareas of distributed algorithms, and they have received extensive attention over the last three decades. See the book by Barenboim and Elkin\cite{barenboim2013monograph} for an instructive survey. 
%These problems can be divided into two categories: \emph{vertex-based problems} such as vertex coloring, maximal independent set, and ruling sets, and \emph{edge-based problems} such as edge coloring, maximal matching, orientations, and forest decompositions. In this paper, our focus is on the second category. 
In this paper, we revisit some of the classical problems of this area, as well as some newer ones which have received attention only recently. The common denominator of the problems we consider is that they all revolve around a seemingly rudimentary edge symmetry-breaking task, which we refer to as \emph{degree-splitting}, and discuss shortly.

Throughout, we work with the standard distributed model called $\mathsf{LOCAL}$, due to Linial\cite{linial1992locality}: The network is abstracted as a graph $G=(V, E)$. There is one processor on each vertex, which initially knows only its neighbors. Per round, each processor can send one message to each of its neighbors.

\vspace{-7pt}
\subsection{Degree Splitting and Edge Coloring}
\vspace{-5pt}
We start our discussion of degree splitting with the classical edge-coloring problem, as these two problems have a simple and clear relation. One basic version of degree splitting is as follows:

\begin{mdframed}[hidealllines=false,backgroundcolor=gray!10]
%\vspace{-2pt}
\textbf{Undirected Degree Splitting:} Given a graph with maximum degree $\Delta$, color each edge red or blue such that each node has at most $\Delta(1+\eps)/2$ edges in each color, for a small $\eps\geq0$.
%\vspace{-2pt}
\end{mdframed}
\vspace{5pt}

To understand the connection to edge-coloring, imagine an ideal world---though not always feasible---where everything is fair and one could always have a perfect split, i.e., with $\eps=0$. Then, by a recursion of depth $\log \Delta$ of degree-splittings, we would get to a setting where edges are colored in $2^{\log \Delta} = \Delta$ colors, and each node has at most $1$ edge of each color. That is a proper $\Delta$ edge-coloring. 

But that is too good to be true! Not every graph admits a $\Delta$ edge-coloring\footnote{Although, a $(\Delta+1)$-edge-coloring is guaranteed to always exist, by Vizing's theorem\cite{vizing1964estimate}.}, and not every graph admits a perfect degree-split. Take $K_3$ for instance. But almost perfect splits exist, with $\eps=2/\Delta$. Here is one way of getting them: Add a dummy node and connect it to all odd degree nodes; then take an Eulerian tour and color its edges red and blue in alternating order. The same guarantee also follows from the Discrepancy Theory result of Beck and Fiala\cite{beck1981integer}---which is in fact far more general\footnote{It follows from Beck-Fiala's theorem that hypergraphs of rank $t$---i.e., where each hyperedge has at most $t$ vertices---admit a red-blue hyperedge-coloring where each node has at most $\Delta/2+t-1$ edges in each color.}. 

We present distributed algorithms for computing almost perfect degree-splits. As Eulerian tours cannot be computed distributedly, our methods are vastly different from the above. As formalized in \Cref{lem:randomizedundirected}, we present a randomized algorithm that achieves an almost-perfect split with per-color degree at most $\Delta/2+1$ in $\poly(\Delta \log n)$ rounds. However, as far as we know, this result itself does not lead to an improvement in any of the well-studied distributed problems. Instead, as formalized in \Cref{thm:coloring}, we present a deterministic algorithm that in $\poly\log n$ rounds produces a degree-split with $\eps = 1/\log^c n$, for a desirably large constant $c\geq 2$, in graphs of max-degree $\Delta =\Omega(\poly\log n)$. This immediately leads to an improvement for edge-coloring: 

\begin{theorem}\label{thm:edge-coloring} There is a deterministic $\poly\log n$-round algorithm for $(2\Delta-1)(1+o(1))$-edge-coloring. 
\end{theorem}
As Panconesi and Rizzi\cite{panconesi-rizzi} state, four key problems of the area from the 1990's were to find $\poly\log n$-round deterministic algorithms for Maximal Independent Set, $(\Delta+1)$-vertex-coloring, $(2\Delta-1)$-edge-coloring, and Maximal Matching. To this day, only the Maximal Matching problem is resolved, due to a breakthrough of Hanckowiak, Karonski, and Panconesi\cite{hanckowiak1998MMConference, hanckowiak2001MMJournal}. \Cref{thm:edge-coloring} almost settles the edge-coloring problem, modulo the $(1+o(1))$ factor. The previously best-known number of required colors were $\Delta\cdot 2^{O(\frac{\log \Delta}{\log\log \Delta})}$, due to Barenboim and Elkin\cite{Barenboim:edge-coloring}, and $O(\Delta\log n)$, due to Czygrinow et al.\cite{czygrinow2001coloring}. We in fact present a simpler proof of the former via degree splittings in \Cref{app:SimpleColoring}. As stated in the 4th open problem of the \emph{Distributed Graph Coloring} book by Barenboim and Elkin\cite[Section 11]{barenboim2013monograph}, even $\Delta \cdot \poly(\log \Delta)$-edge-coloring in $\poly\log n$ rounds remained open.

In \Cref{sec:randomizedcoloring}, we explain that mixing \Cref{thm:edge-coloring} with some other ideas also leads to a fast randomized $(2\Delta-1)(2+o(1))$-edge-coloring in $O(\poly(\log\log n))$ rounds.

\subsection{Degree Splitting and Edge Orientations}
\vspace{-2pt}
We also consider the following natural variant of degree-splitting:
\vspace{3pt}
\begin{mdframed}[hidealllines=false,backgroundcolor=gray!10]
\textbf{Directed Degree Splitting:} Given a graph with maximum degree $\Delta$, orient each edge such that each node has in-degree and out-degree at most $\Delta/2(1+\eps)$, for a small $\eps\geq 0$.
\end{mdframed}
\vspace{5pt}

We note that in directed splitting, even in graphs which admit perfect splits, computing them might be time-consuming---e.g., perfectly splitting an $n$-node cycle would require $\Omega(n)$ rounds.
   
The directed degree splitting relates closely to Eulerian Orientations, as the latter requires an orientation such that each node has the same in-degree and out-degree. We note that on bipartite graphs (with the bi-partition given), the directed and undirected degree-slitting problems are equivalent. However, we are not aware of a formal reduction between them in the general case. 

Despite that, we use more or less the same methodology as that of our undirected degree splitting to find almost perfect directed splits: As formalized in \Cref{lem:randomizeddirected}, we get an existentially best possible guarantee of per-node out-degree and in-degree at most $\Delta/2+1$ in $\poly(\Delta \log n)$ rounds randomized. More importantly, as formalized in \Cref{lem:largeorient}, we get a deterministic algorithm that computes a degree splitting with $\eps = 1/\log^c n$, for a desirably large constant $c\geq 2$, in graphs of max-degree $\Delta =\Omega(\poly\log n)$ in $\poly \log n$ rounds.

\paragraph{Sinkless Orientation and Its Refinements} One related problem which has received attention recently is \emph{sinkless orientation}, where the objective is to orient edges of a $\Delta$-regular graph such that each node has out-degree at least $1$. Note that this is clearly a much weaker requirement than that of the directed degree splitting problem. Recently, Brandt et al.\cite{brandt2016LLL} gave an elegant $\Omega(\log_{\Delta} \log n)$ round lower bound for this problem, which was then extended by Chang et al.\cite{chang2016exponential} to an $\Omega(\log_{\Delta} n)$ lower bound for deterministic algorithms. For this weaker orientation problem, we can achieve much better round complexities, which match the respective lower bounds.

\begin{theorem}\label{thm:sinkless} There is a randomized $O(\log_{\Delta} \log n)$-round algorithm which solves sinkless orientation in all $\Delta$-regular graphs, for $\Delta\geq 3$, and in fact all graphs of minimum degree at least $\Delta$, with high probability. Moreover, the same problem has a deterministic $O(\log_{\Delta} n)$-round algorithm.
\end{theorem}
We can in fact guarantee a more balanced split in almost the same running time: We show how to find an orientation with per-node in-degree and out-degree at most $5\Delta/6$---that is, a directed degree split with $\eps=2/3$---in $O(\log\log n)$ rounds randomized and $O(\log n)$ rounds deterministically. 

Brandt et al.\cite{brandt2016LLL} used sinkless orientation to prove an $\Omega(\log_{\Delta}\log n)$ round lower bound on $\mathsf{LOCAL}$-algorithms for Lovasz Local Lemma (LLL). Since LLL can provide much finer degree splits, studying these stronger degree-splits might expose higher lower bounds for LLL. Moreover, Chang et al.\cite{chang2016exponential} recently presented the first exponential separation between randomized and deterministic distributed complexities, by showing that $\Delta$-vertex-coloring $\Delta$-regular trees requires $\Theta(\log_{\Delta}\log n)$ rounds randomized and $\Theta(\log_{\Delta} n)$ rounds deterministically. \Cref{thm:sinkless}, in conjunction with the aforementioned lower bounds, exhibits the same exponential separation on sinkless orientation.

\paragraph{Low Out-Degree Orientations and Nash-Williams Decompositions} An alternative way of viewing the directed degree splitting problem is as follows: it asks for an orientation that achieves a maximum out-degree within a $1+\eps$ factor of what is necessary, given the maximum degree. In this regard, one can ask for a stronger guarantee: find an orientation with maximum out-degree $a(1+\eps)$, where $a$ is the arboricity of the graph. Clearly any orientation of any graph with arboricity $a$ needs out-degree at least $a$. Our methods allow us to achieve such an approximation: 

\begin{theorem} There is a randomized $O(\poly \log n/\eps)$-round algorithm which orients $a$-arboricity graphs with maximum out-degree at most $a(1+\eps)$. This is equivalent with a decomposition into $a(1+\eps)$ pseudo-forests. If $a=\Omega(\log n)$, we can turn this into a decomposition into $a(1+\eps)$ forests.
\end{theorem}
See \Cref{thm:directedDegreeSplit} and \Cref{lem:forestDecomp} for the formal statements. We note that efficient distributed orientation with out-degree less than $2a$ had remained open. The best previously known results were as follows: an orientation with out-degree at most $2a$ in $O(a\log n)$ rounds and an orientation with out-degree at most $(2+\eps)a$ in $O(\log n)$ rounds, both due to Barenboim and Elkin\cite{barenboim2010sublogarithmic}. The same authors state the closely-related problem of efficient distributed decomposition into less than $2a$ forests as the 10th open problem in their book\cite[Section 11]{barenboim2013monograph}.

\vspace{-7pt}
\subsection{Other Related Work}
\vspace{-8pt}
\paragraph{Related Work on Degree Splitting} Degree splitting was first considered by Israeli and Shiloach\cite{israeli1986improved} in the parallel algorithms model---a.k.a.~PRAM---as a subroutine for computing a maximal matching. Their method for computing the split relies on finding an Eulerian cycle of the graph and $2$-coloring its edges in alternating order. We note that a number of other work in the PRAM model, e.g., \cite{karloff1987efficient}, also used degree splittings but all relying on Eulerian cycles. Unfortunately, Eulerian cycles cannot be computed using distributed algorithms (with sublinear complexity). Hence, these methods cannot extend to our setting. Our approach is quite different, and it is morally much closer to the classic ideas of augmenting paths and blocking-flows in the maximum flow problem\cite{edmonds1972theoretical, dinicalgorithm}.

However, there is an ingenious line of work that comes close to distributedly computing degree splits. Inspired by the parallel maximal matching algorithm of Israeli and Shiloach\cite{israeli1986improved}, Hanckowiak, Karonski, and Panconesi\cite{hanckowiak1998MMConference, hanckowiak2001MMJournal} used a clever relaxation of the degree splitting to distributedly compute a maximal matching in $\poly\log n$ rounds deterministically. Their relaxation allows a small but non-negligible fraction $\delta>0$ of nodes to have an unfair split, even possibly having all of their edges in one color. This relaxation is indeed essential to their method. While for Maximal Matching this relaxation is good enough, it becomes too costly for the other problems, e.g., the number of required colors in edge-coloring blows up by an $O(\log n)$ factor, as Czygrinow et al.\cite{czygrinow2001coloring} show. In this regard, one can view our degree-splitting results as a \emph{qualitative} improvement on those of \cite{hanckowiak1998MMConference, hanckowiak2001MMJournal}, as we do not need the relaxation of admitting some unbalanced nodes. However, we pay for this refinement and our round-complexity ends up being a higher $\poly\log n$.  

\paragraph{Related Work on Edge Coloring}
Edge coloring is one of the classical distributed problems and it has been studied extensively over the years. There is a clear dichotomy between randomized and deterministic algorithms for this problem, and we review the results in these two categories.

First, we review the deterministic results. Panconesi and Rizzi provide an $O(\Delta+\log^* n)$ algorithm for $(2\Delta-1)$-edge-coloring. This complexity was recently improved by Barenboim\cite{Barenboim:2015} to $O(\Delta^{3/4}\log \Delta + \log^* n)$ and subsequently to $O(\Delta^{1/2}\log^{5/2} \Delta + \log^* n)$ by Fraigniaud, Heinrich, and Kosowski~\cite{Fraigniaud:2016}. Both these results work indeed for the harder problem of $(\Delta+1)$-vertex coloring. However, these complexities can be much larger than $\poly\log n$. The best known number of required colors for $\poly\log n$-round algorithms remained at $\Delta\cdot 2^{O(\frac{\log \Delta}{\log\log \Delta})}$, due to Barenboim and Elkin\cite{Barenboim:edge-coloring}, and $O(\Delta\log n)$, due to Czygrinow et al.\cite{czygrinow2001coloring}. See also \cite[Chapter 8 \& Chapter 11.1]{barenboim2013monograph}.

Now, we review the randomized results. The classical $O(\log n)$ round randomized Maximal Independent Set algorithm of Luby\cite{luby1986simple} leads to a randomized $O(\log n)$-round $(2\Delta-1)$-edge coloring. This round complexity was improved to $O(\log \Delta) + e^{O(\sqrt{\log\log n})}$ by Barenboim et al.\cite{barenboim2012locality}. This was further improved to just $2^{O(\sqrt{\log\log n})}$ by Elkin, Pettie, and Su\cite{EPS15}. In the randomized world, there are also algorithms for finding colorings with smaller number of colors. Panconesi and Srinivasan\cite{panconesi1997randomized} give the first such result. Dubhashi, Grable and Panconesi\cite{dubhashi1998near} later improve this to $(1+\eps)\Delta$-edge-coloring in $O(\log n)$ rounds, when $\Delta=\Omega(\log^{1+\Omega(1)} n)$. This was later refined and extended to graphs of degree $\Delta\geq C_0$, for a constant $C_0$ depending on $\eps$. The final work in this track is by Elkin, Pettie, and Su\cite{EPS15} which improves the complexity to $O(\log^* \Delta \cdot \max\{1,\frac{\log n}{\Delta^{1-o(1)}}\})$.

\paragraph{Related Work on Orientations}
Distributed low out-degree orientation of low-arboricity graphs was first studied by Barenboim et al.\cite{barenboim2010sublogarithmic} and the same results have been used in a few subsequent works. This orientation was then turned into a forest decomposition which subsequently lead to sublinear-time algorithms for maximal independent set, vertex coloring, edge coloring and maximal matching in graphs of low arboricity. See \cite[Chapter 4 \& Chapter 11.3]{barenboim2013monograph}. Sinkless orientation was recently introduced by Brandt et al.\cite{brandt2016LLL} and studied also by Chang et al.\cite{chang2016exponential}.

\subsection{Our General Method In a Nutshell}
Our general method follows a natural idea and it is inspired by classical concepts from the maximum flow problem. Although, as we soon see, to have an efficient algorithm, particularly an efficient distributed algorithm and especially a deterministic one, various aspects require novel techniques.

Let us consider directed degree splitting as our running example in discussing the methodology. Consider an arbitrary orientation. In this orientation, some nodes might have an out-degree (much) larger than in-degree and some nodes might have a larger in-degree than out-degree. Virtually, we can think of out-degree as the commodity of our flow. This means that the first group have \emph{excess flow} and the second group have \emph{flow deficiency}. Naturally, we wish to transfer some flow from the first group to the second to even things out. For instance, if we find a directed path from the first group to the second---i.e. what we usually call an \emph{augmenting path}---we can flip the direction of its edges, improving the degrees on the two endpoints, but keeping them unchanged in the middle nodes. We would continue doing this until each node has about the same out-degree and in-degree, at which point we have found our almost-perfect degree split. 

Finding an efficient distributed algorithm following this idea necessitates a number of considerations and novel techniques. First, we need the augmenting paths to be short, as otherwise we cannot even find them distributedly. That issue is not hard, because as we will see imperfect splits have relatively short augmenting paths. Second, in a fast distributed algorithm, we cannot afford to fix imbalanced nodes by using augmenting paths one by one; we instead need to have many ``disjoint" short augmenting paths. Third, we need to find them fast distributedly. These second and third issues are much more crucial. We will show that imperfect splits in fact have large sets of ``disjoint" augmenting-paths. However, finding such a set distributedly, and especially doing it deterministically, will require quite some effort. We leave those discussions to the technical sections. 

Finally, we note that in some problems, we desire much faster solutions, e.g., an $O(\log_{\Delta}\log n)$ round complexity in sinkless orientation. In these cases, our general methodology provides some algorithm but not quite matching the lower bound. There, we deviate slightly from this flow augmenting mentality and use some other ideas to optimize the complexity. 

%\subsection{Organization}
%\mtodo{add a raodmap, after cutting down}
%\newpage
\section{Sinkless Orientation}\label{sect:sinkless}
In this section, we present a simple $O(\log_{\Delta} \log n)$ round randomized algorithm for sinkless orientation in $\Delta$-regular graphs. This matches the $\Omega(\log_{\Delta} \log n)$ lower bound presented by Brandt et al.\cite{brandt2016LLL}. As a component of this result, we also present an $O(\log_{\Delta} n)$ round deterministic algorithm for the same problem, which itself proves the corresponding $\Omega(\log_{\Delta} n)$ lower bound of Chang et al.\cite{chang2016exponential} to be tight. These results prove \Cref{thm:sinkless}. We first present an algorithm that works assuming $\Delta>500$. The extension to all cases with $\Delta\geq 3$, to cases with irregularity, and also to a setting where we desire a more balanced split of in-degree and out-degree, are deferred to \Cref{app:sinklessGeneralizations}. 

\subsection{Sinkless Orientation for $\Delta$-regular graphs with $\Delta>500$}
\label{sec:sinklessLargeDeg}
Notice that if $\Delta=\Omega(\log n)$, orienting each edge randomly ensures that all nodes have at least $\Theta(\Delta)$ outgoing edges, with high probability, which is a sinkless orientation. The far more interesting case of the problem is when $\Delta = O(\log n)$, and this is the focus of this section. 

The algorithm is composed of two parts: a randomized part, which \emph{shatters} the graph thus leaving only small components of $\polylog n$ size, and then a deterministic part, which takes care of these remaining small components. The random part will take $O(1)$ rounds and then the deterministic part will solve these remaining $\polylog n$-size components in $O(\log_{\Delta} \log n)$ rounds, hence leading to the overall randomized complexity of $O(\log_{\Delta} \log n)$. The deterministic part itself is a full solution for $n$-node graphs with complexity $O(\log_{\Delta} n)$. Next, we present these two parts.
\subsubsection{The Randomized Part of the Algorithm (Pre-Shattering)}\label{sec:randomSinkless}
The randomized algorithm is quite simple and it works as follow. Before presenting the algorithm, we note that in the course or the algorithm, we will allow \emph{half-edges}, which are edges with only one endpoint (which needs the edge to be outgoing). 
\begin{algorithm}[H]
\caption{Randomized Orientation}\label{alg:random}
\begin{algorithmic}
\small
\State Mark each edge with probability $\frac{1}{4}$.
\State For each marked edge, orient it randomly with probability $1/2$ for each direction.
\State For each node $v$, mark $v$ as {\it a bad node} of the following types according to these rules:
\begin{itemize}
	\item Type I. If $v$ has more than $\Delta/2$ marked edges incident to it.
	\item Type II. If $v$ is not Type I but it has at least one neighbor of Type I.
	\item Type III. If $v$ is not Type I or Type II but it has no outgoing marked edges.
\end{itemize}
\State Unmark all the edges incident to Type I nodes.
\State Orient unmarked edges which both of their endpoints are good nodes arbitrarily.
\State Consider unmarked edges with exactly one good endpoint as a half-edge only attached to the bad-node.
\State Run the deterministic algorithm on the components induced by the bad nodes and their edges or half-edges.
\end{algorithmic}
\end{algorithm}
The analysis is presented in \Cref{app:sinklessLargeDeg}. Here, we only mention a few key aspects which should deliver the intuition. One can see that, after \Cref{alg:random}, each bad node is incident to at least $\Delta/2$ unmarked edges or half-edges.  Moreover, each good node has at least one outgoing marked edge. Also, the probability of each node to be bad is at most $exp(-\Delta)$. As we show in \Cref{crl:componentSize}, it follows that with high probability, each connected component induced by bad nodes has size $O(\Delta^2 \log n) = \poly\log n$. The deterministic algorithm solves these remaining components, where each remaining node has at least $\Delta/2$ edges or half-edges, in $O(\log_{\Delta}\log n)$ rounds.

\subsubsection{The Deterministic Part of the Algorithm (Post-Shattering)}\label{sec:deterministic}
Consider an $N$-node graph where each node is incident on at least $d\geq 3$ edges or half-edges. We explain how to find a sinkless orientation of this graph in $O(\log_{d} N)$ rounds deterministically. Note that when plugging this subroutine in the algorithm of the previous section, we will have $N=\poly\log n$ and $d\geq \Delta/2$. Hence, this deterministic piece would work in  $O(\log_{\Delta} \log n)$ rounds. 

\paragraph{The Deterministic Algorithm} Orient half-edges outwards from their single endpoint. For the edges, do as follows: Uniquely identify cycles by appending the ids of the related edges. For each cycle, define its preferred orientation by taking the smallest id edge from the lower id node to the higher id node, and then following this direction through the whole cycle. Call a cycle short if it has at most $2\log_{d-1} N+1$ edges. Call an edge \emph{short} if it is in at least one short cycle. First, orient each short edge $e$ consistent with the preferred orientation of the smallest id short cycle that contains $e$. Call a node \emph{short} if it is incident on at least one short edge, or on a half-edge, and \emph{long} otherwise. Then, for each long node $v$, let $u$ be one neighbor of $v$ who is closer to short nodes (compared to $v$). Then orient this edge as $v \rightarrow u$.

\begin{lemma}\label{lem:deterministicSinkless} For any $d\geq 3$, this deterministic algorithm works in $O(\log_{d} N)$ rounds on any $N$-node graph where each node is incident on at least $d$ edges or half-edges, and it orients these edges or half-edges such that each node has out-degree at least $1$.
\end{lemma}
\begin{proof}Regarding the time complexity, notice that each half-edge is oriented immediately and each short edge can find its orientation in $O(\log_{d-1} N)=O(\log_{d} N)$ rounds. Hence, in $O(\log_{d-1} N)$ rounds, we have all short edges oriented. We argue that orienting edges for long nodes can also be done in  $O(\log_{d-1} N)$ rounds because each such node has distance at most $O(\log_{d-1} N)$ to some short node. More concretely, we argue that each node has either a half-edge or a short cycle within its $\log_{d-1} N$ neighborhood. This is true because otherwise, the BFS tree of depth $\log_{d-1} N$ rooted at this node would have minimum degree $d$ and depth $\log_{d-1} N$. Such a tree necessarily has more than $N$ nodes, which would be a contradiction. 

We next argue that each node has out-degree at least $1$. The argument for long nodes is easy as they are oriented towards the short cycles. The argument for nodes incident on half-edges is also trivial. The key part is to argue that in the orientation of short edges, despite the fact that different short edges act according to possibly different short cycles, each short node has out-degree at least $1$. For that, consider a short node $v$, consider all the short edges incident on it, and let $C$ be the smallest id short cycle which contains at least one of these edges. Suppose that $C$ has edges $e_1=(v, u_1)$ and $e_2=(v,u_2)$ incident on $v$. If $e_1$ is oriented as $v\rightarrow u_1$, we are done. Suppose $e_1$ is oriented as $u_1\rightarrow v$. We claim that then it must be the case that $e_2$ is oriented as $v \rightarrow u_2$, which would finish the proof. This claim is true because the only reason for $e_2$ to be oriented in the opposite direction is if $e_2$ is a part of short cycle $C'$ which has an id smaller than that of $C$. However, that would be in contradiction with the choice of $C$.
\end{proof}

\section{Edge-Coloring via Undirected Degree Splitting}\label{sec:undirected}
In this section, we explain a method for $(2\Delta-1)\cdot(1+o(1))$-edge coloring graphs with maximum degree $\Delta$, based on degree splitting. As a formalized restatement of \Cref{thm:edge-coloring}, we get: 

\begin{theorem}\label{thm:coloring}There is a deterministic distributed algorithm that computes a $(2+\epsilon)\Delta$-edge coloring of any graph with maximum degree $\Delta$, in $O(\log^{11} n / \epsilon^3 )$ rounds. 
\end{theorem}

We note that, coarser degree splittings can also be used to obtain much simpler algorithms for edge coloring. Particularly, in \Cref{app:SimpleColoring}, we explain a much simpler method for edge-coloring that matches the bounds of Barenboim and Elkin \cite{Barenboim:edge-coloring}. \fullOnly{In the edge-coloring algorithm provided in the previous section, the Achilles heel is in the $2$ factor that we lose in each recursion level for splitting. After a number of iterations, these $2$ factors accumulate and amount to a great loss in the number of colors.} However, to get \Cref{thm:coloring}, we need to have an almost perfect split, particularly the loss in each degree-splitting iteration should be at most a $1+{1}/{\poly\log n}$. This would allow us to say that even after $\log n$ iterations, the overall loss is negligible. 

In the rest of this section, we explain how to achieve this fine degree splitting. We first explain in \Cref{subsec:low-degree} how to split graphs of maximum degree at most $\poly\log n$ into two spanning subgraphs, each with maximum degree almost half the previous maximum degree. Then, in \Cref{subsec:high-degree}, we explain how to lift this solution to graphs of higher degree, and how to use that to obtain the claimed edge-coloring result. %In \Cref{sect:randomundirected}, we give a randomized algorithm that is capable of achieving the best possible split. Finally, in \Cref{sec:randomizedcoloring}, we show how to combine the graph shattering technique and our deterministic algorithm to obtain a faster randomized edge-coloring algorithm.

\subsection{Deterministic Undirected Degree Splitting for Low-Degree Graphs}\label{sec:lowdeg}
\label{subsec:low-degree}

We say a red-blue edge coloring is $t$-balanced, if there are at most $t$ red edges and at most $t$ blue edges incident to each node.  We show the following result:
\begin{lemma}\label{lem:smallbalance}Given a graph $G$ with maximum degree $d$, a $\lfloor (1+\epsilon)d/2 \rfloor$-balanced coloring can be obtained in $O((d^2 \log^5 n) /\epsilon)$ rounds provided that $(4\log_{1.5}m) / d < \epsilon < 1$. \end{lemma}

\begin{figure}
\centering

\begin{subfigure}[t]{0.47\textwidth}
\includegraphics[scale = 0.6]{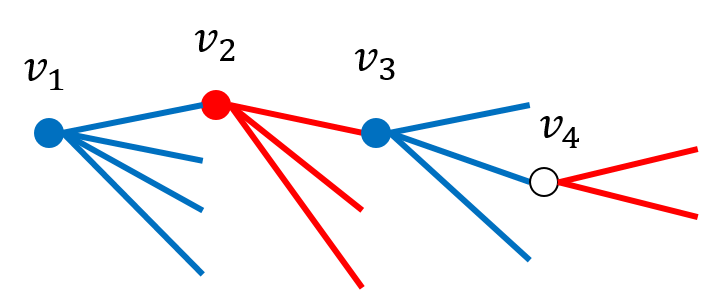}\caption{Before the augmentation.}
\end{subfigure}
\qquad
\begin{subfigure}[t]{0.47\textwidth}
\includegraphics[scale = 0.6]{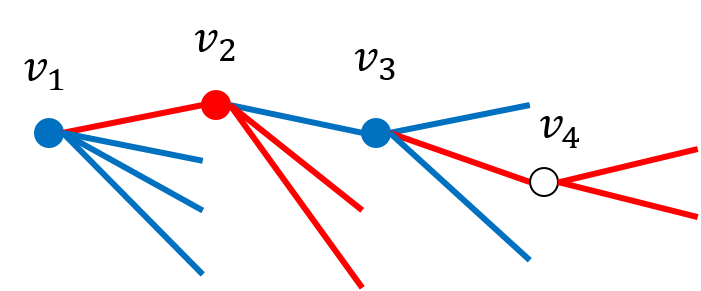}\caption{After the augmentation.}
\end{subfigure}
\caption{\label{fig:augmenting_path} $t = 4$ and $v_1 v_2 v_3 v_4$ forms an augmenting path.} %The blue nodes, red nodes, and other nodes are denoted by solid nodes, empty nodes, and gray nodes respectively. The thick edges denote the blue edges and the thin edges denote the red edges.}
\end{figure}

Given a $t$-balanced coloring with $t> \lfloor (1+\epsilon)d/2 \rfloor$, we show how to improve the coloring to a $(t-1)$-balanced coloring. Then we  iterate $t$ from $d$ to $\lfloor (1+\epsilon)d/2 \rfloor+1$. 

Label a node blue if it has $t$ or $t-1$ blue edges incident to it. Label a node red if it has $t$ or $t-1$ red edges incident to it. A node is a {\it source} if it is incident to  exactly $t$ red edges or $t$ blue edges. Let $S$ be the set of all source nodes. An {\it alternating path} $v_1 \ldots v_{k}$ is a path that satisfies the following:
\begin{enumerate*}
	\item $v_1 \in S$.
	\item $v_i$ and $v_{i+1}$ alternate between red and blue for $1 \leq i \leq k-2$.
	\item Edge $v_i v_{i+1}$ is colored the same with the label of $v_i$ for $1 \leq i \leq k-1$.
\end{enumerate*}
An {\it augmenting path} is an alternating path with an additional condition:
\begin{enumerate*} \setcounter{enumi}{3}
	\item \label{item:last} $v_k$ is unlabeled or labeled differently than $v_{k-1}$.
\end{enumerate*}

Suppose there is an augmenting path that starts at $v_1$. Without loss of generality, suppose that $v_1$ is blue, $v_1 v_2$ is blue, $v_2$ is red, $v_2 v_3$ is red, etc. By augmenting along the augmenting path, we recolor $v_1 v_2$ with red, $v_2 v_3$ with blue, etc. The number of red/blue edges incident to $v_2, v_3, \ldots v_{k-1}$ remain the same after the augmentation. The number of blue edges incident to $v_1$ decreases by 1. If $k$ is even, then $v_{k-1} v_k$ was blue. When we recolor $v_{k-1} v_k$ in red, the number of red edges incident to $v_k$ increases by 1. Since $v_k$ is not labelled red by (\ref{item:last}.), after the increment, the number of red edges is still at most $t-1$. Therefore, the number of source nodes decreases by 1. If $k$ is odd, then a similar argument applies. See \Cref{fig:augmenting_path} for an illustration.

Thus, given an augmenting path, we can decrease the number of source nodes by 1. By finding the augmenting paths repeatedly, we can eliminate all the sources, and so the graph becomes $(t-1)$-balanced. However, to reduce the number of sources efficiently, we need to do multiple augmentations in parallel. Therefore, we find a set of {\it almost edge-disjoint augmenting paths}, as described below. 

\paragraph{Finding Many Almost Edge-Disjoint Augmenting Paths}
\begin{definition}[Ordered Disjointness]
We say two paths $v_1 \ldots   v_k$ and $v'_1  \ldots v'_{k'}$ are {\it ordered disjoint} if there does not exist  $1 \leq i < k$ and $1 \leq j < k$ such that the ordered pair $(v_i, v_{i+1})$ and $(v'_j,v'_{j+1})$ are the same.
\end{definition}
 
     A set of almost edge-disjoint augmenting paths is a set of augmenting paths with the following property (see \Cref{fig:almostdisjnt} for an example):
\begin{enumerate*}
\item \label{itm:almt_disjoint1} The first nodes of the augmenting paths in $S$ are distinct.
%\item \label{itm:almt_disjoint2} No prefix of the augmenting path is an augmenting path.
\item \label{itm:almt_disjoint3} For any two augmenting paths in $S$, they are ordered disjoint.
\end{enumerate*}

The second property actually characterizes the edge-disjointness property among the paths except on the last edge. Indeed, suppose that $(v_i, v_{i+1})$ and $(v'_j, v'_{j+1})$ denote the same edge $v_i v_{i+1}$ (that is, $\{v_i, v_{i+1} \}= \{v'_{j}, v'_{j+1}\}$). Then $v_i$ and $v'_j$ must have been labeled the same color with the edge $v_i v_{i+1}$ by definition of an augmenting path. Moreover, since $v_i \neq v'_j$, we must have $v_{i} = v'_{j+1}$ and $v'_{j} = v_{i+1}$. This implies that $v_{i}$ and $v_{i+1}$ are labeled the same color and so are $v'_{j}$ and $v'_{j+1}$. Therefore, $v_1  \ldots v_{i+1}$ and $v'_1 \dots v'_{j+1}$ must be augmenting paths, which implies that $v_i v_{i+1}$ is their last edge.

The augmenting paths  may share the same last node. Augmenting along more than one augmenting paths of the same last node may increase the red/blue degree of it to more than $t-1$, creating new sources. However, if each such node accepts only one path that ends at it and we only augment the accepted paths, then the red/blue degrees can only increase to at most $t-1$. Since we know that there can be at most $d$ augmenting paths that end at the same node, if we find $\Omega(|S|)$ augmenting paths, at least $1/d$ fraction of them can be augmented. Thus, the sources can be reduced to $(1-\Omega(1/d))\cdot |S|$. 

In the rest of this section we show how to find $\Omega(|S|)$ almost edge-disjoint augmenting paths deterministically with the restriction that $(4\log_{1.5}m) / d < \epsilon < 1$. In \Cref{sect:randomundirected} we show how to find $\Omega(|S|)$ almost edge-disjoint augmenting paths randomly but without the restriction that $\epsilon d = \Omega(\log n)$.

\paragraph{Technical Overview} \Cref{cor:maximalalmostdis} in \Cref{sect:randomundirected} shows that there exists a set of $|S|$ almost edge-disjoint augmenting paths of length $O(\log n /\epsilon)$. Moreover, the cardinality of any maximal set of almost edge-disjoint augmenting paths of length $O(\log n /\epsilon)$ is $|S|$. This implies in the sequential setting we can find such a set of paths greedily. Similar to the approach in \cite{LotkerPP-distrmatch} by Lokter et al., in \Cref{sect:randomundirected}, we show it can be found in the distributed setting randomly by simulating Luby's maximal independent set algorithm on a super-graph. However, finding a maximal set of (almost) edge-disjoint augmenting paths deterministically  is more challenging technically. Fortunately, in this problem, we have the property that if we build a search tree from a source to search for an augmenting path, the search tree expands very quickly. This property allows us to build (almost) edge-disjoint search trees from different sources such that they are still able to expand quickly. Note that our algorithm and analysis below is independent of that in \Cref{sect:randomundirected}.

\begin{figure}
\centering

\begin{subfigure}[t]{0.47\textwidth}
\includegraphics[scale = 0.6]{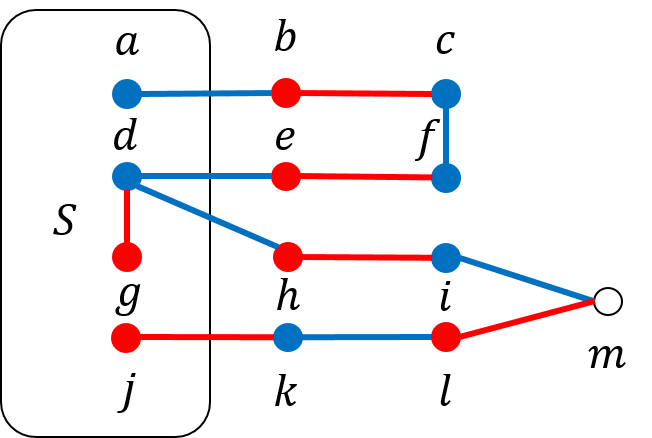}\caption{A set of almost edge-disjoint augmenting paths. $S$ is the set of sources. The 4 augmenting paths are $abcf$, $defc$, $gdhim$, and $jklm$. Note that $cf$ is included in both $abcf$ and $defc$, as their last edge.}\label{fig:almostdisjnt}
\end{subfigure}
\qquad
\begin{subfigure}[t]{0.47\textwidth}
\includegraphics[scale = 0.6]{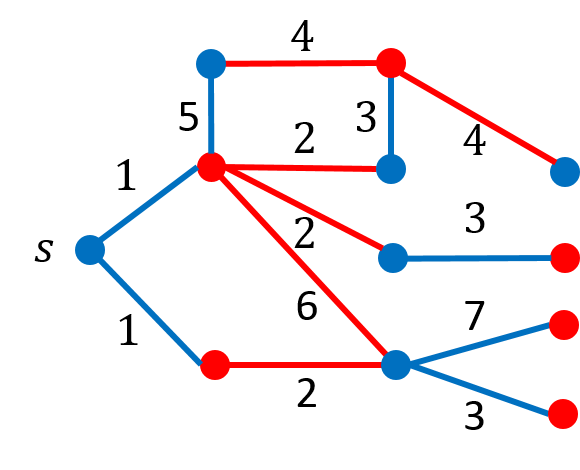}\caption{A pseudo-tree  rooted at $s$ with 4 leaves. The numbers on the edges denote their depth in the pseudo-tree. Each leaf corresponds to an alternating path from $s$ to it. }\label{fig:pseudotree}
\end{subfigure}
\caption{\label{fig:multple_paths} A set of almost edge-disjoint augmenting path and a pseudo-tree.}
%\caption{\label{fig:multple_paths} The blue nodes, red nodes, and other nodes are denoted by solid nodes, empty nodes, and gray nodes respectively. The thick edges denote the blue edges and the thin edges denote the red edges.}
\end{figure}

The deterministic algorithm is described in \Cref{alg:augment}. The main idea is to grow a tree-like structure from each $s \in S$ simultaneously. We call these {\it pseudo-trees} (See \Cref{fig:pseudotree} for an example). In a pseudo-tree, each edge is associated with 0,1, or 2 children edges, who are the adjacent edges. Each edge except those who are adjacent to $s$ have exactly one parent edge. The edges without children are leaves. Also, the structure does not contain a cycle. That is, if we view each edge as a node, then the pseudo-tree is a collection of rooted trees whose roots are the edges incident to $s$. Thus, there is an unique path from $s$ to each leaf. We require each such path to be an alternating path.

When we grow the pseudo-trees $T_s$ and $T_{s'}$ from different sources $s$ and $s'$, we require that any alternating paths $v_1 \ldots v_k$ from $T_s$ and $v'_1 \ldots v'_{k'}$ from $T_{s'}$ are ordered disjoint. This ensures that if we extract an augmenting path from each of the pseudo-trees, they will be almost edge-disjoint.

%First, an alternating path from the source $s$ to a leaf may visit the same node twice (but never the same edge twice). For this reason, the structure we maintain for each source $s\in S$ is no longer a tree, but we call it a {\it pseudo-tree}. In a pseudo-tree, each path from the source to the leaf is an alternating path. Moreover, for any two pseudo-trees $T_s$ and $T_{s'}$, any alternating paths $v_1 \ldots v_k \in T_s$ and $v'_1 \ldots v'_{k'} \in T_{s'}$ are ordered disjoint. This ensures that if we extract an augmenting path from each of the pseudo-trees, they will be almost edge-disjoint.

To grow the pseudo-trees simultaneously, each source $s \in S$ maintain a set of tokens. Initially, there is only one token starting at each source $s \in S$. Then, they split as they travel. The edges traveled by the tokens from $s$ form a pseudo-tree of $s$. In each step, the tokens at each node request for the edges to grow the pseudo-trees. Because of the ordered disjointness property, each edge $uv$ can be assigned only once to the tokens at $u$ and once to the tokens at $v$. In our algorithm, each request of the token will be granted either 1 edge or 2 edges. Those who got 1 edge will travel along the edge. Those who got 2 edges will split into two tokens. Then each edge will be traveled by one of the token. We will assign the edges properly so that a large fraction of the pseudo-trees will grow exponentially. 

\begin{algorithm}
\caption{Finding Almost Edge-Disjoint Augmenting Paths($S$)}\label{alg:augment}
\begin{algorithmic}\small
\State Each source $s \in S$ create a token located at $s$.
\For{level $i = 1,2,\ldots, \log_{1.5} m$}
\State For each node $u$, set its budget, $budget(u)$, to be $\lfloor (2t- d -2)/ \log_{1.5} m \rfloor$.
\For{step $j = 1,2,\ldots, h$, where $h = \lceil (16/3)\log^2_{1.5} m / \epsilon \rceil$}
\State Each active token requests for edge assignment at its current node.
\State For each node $u$, $\min(\#request,budget(u))$ requests  will be granted two tokens, others one edge. Each red (blue) node assign unused, distinct red (blue)  edges to tokens. Update $budget(u)$.
\If{a token has been assigned one edge}
	\State Travel along the edge.
\Else \Comment{The token has been assigned two edges.}
	\State Split into two tokens. Then each token take one of the edge and travel along it.
	\State Deactivate the tokens.
\EndIf
\EndFor
\State For each source:
\If {a path from the source to one of its tokens forms an augmenting path} 
\State Deactivate all the tokens from $s$ and save the augmenting path.
\ElsIf{it has less than $L_{i+1}$} 
\State Mark $s$ as {\bf failure} and remove $s$ from $S$ and  deactivate all the tokens. 
\ElsIf{it has at least $L_{i+1}$}
\State Discard the tokens so that it has exactly $L_{i+1}$ tokens. 
\State Set the $L_{i+1}$ tokens to be active.
\EndIf
\EndFor
\end{algorithmic}
\end{algorithm}

   Define $L_1 = 1$ and $L_{i+1} = 2 \lceil 3 L_{i}/ 4 \rceil$. \Cref{alg:augment} consists of multiple levels. At level $i$, an active source $s$ maintains $L_{i}$ tokens. The goal is for $\lceil 3 L_{i}/ 4 \rceil$ tokens of $s$ to split into two tokens, so the number of tokens becomes $L_{i+1} = 2 \lceil 3 L_{i}/ 4 \rceil$ (discard the rest). Once a token splits into two, they will pause until the next level. We will show that a large fraction of sources achieve the goal in each level.

Each level has $h$ steps. In each step $j$ of level $i$, the active tokens are assigned either 1 or 2 edges. If a node has more budget than the number of requesting tokens, then every token on it is given 2 edges. Otherwise, the node is allowed to assign $budget(u)$ tokens 2 edge; other tokens get 1 edge. % Each edge will not be assigned twice to the tokens at the same node during the entire algorithm. This ensures the ordered disjointness property.

Now we bound the number of requests at $u$ by $d-t+1$. First note that the requests happen only at the labelled nodes, since if a token reaches an unlabelled node then an augmenting has been found. W.l.o.g.~assume $u$ is blue. There are two cases. If $u$ has exactly $t-1$ blue edges, then there is at most $d-t+1$ red edges incident to it. If a token enters through one of the blue edges then the augmenting path must have been found and so the token will be deactivated. Therefore, those who are requesting must enter $u$ through red edges. Thus, the number of requests is at most $d-t+1$. If $u$ has exactly $t$ blue edges, then there is at most $d-t$ red edges incident to it. By the same argument, the number of requests is at most $d-t$. However, since $u \in S$, a additional token has been created at $u$ initially. In any case, the number of requests at $u$ is bounded $d-t+1$. 

The intuition of why each token can split without traveling too far is that because the number of edges that can be assigned is at least $t-1$, the average number of edges that can be assigned to each request is $(t-1)/(d-t+1) = 1 + \Omega(\epsilon)$, since $t > \lfloor (1+\epsilon)d/2 \rfloor$. Therefore, intuitively, it seems possible for each pseudo-tree to grow by an $1+\Omega(\epsilon)$ factor in each step. However, since $1+\Omega(\epsilon)$ is not even an integer, it is not clear what it means to have each token to split into $1+\Omega(\epsilon)$ tokens. Instead, our analysis shows that most tokens will split into two after $\tilde{O}(1/\epsilon)$ steps.
 
We use $budget(u)$ to control how many edges can a node assign in each level.  Suppose that we assign 1 edge to all the $d-t+1$ requests first, then the number of unused blue edges  is at least $t-1 - (d-t+1) = 2t - d -2$. These edges can be used to grant $2t - d -2$ tokens 1 additional edge. We divide the budget across all the $\log_{1.5} m$ levels, so each level has $\lfloor (2t-d-2) / \log_{1.5} m \rfloor$ budget. We say a token did not {\it successfully split}, if it has not been granted two edges during the $h$ steps. We say a source {\it failed} if less than $3/4$ fraction of their tokens sucessfully split. Given the budget, we show that the number of failure sources in each level is bounded. If a source did not fail or did not found an augmenting path yet, then we say it is {\it active}.

\begin{lemma}\label{lem:smalldeg}
Suppose that $(4\log_{1.5}m) / d < \epsilon < 1$. Let $S_i$ denote the set of active sources at the beginning of level $i$ and let $F_i$ be the set of failure sources during level $i$. We have $|F_i| \leq |S_i|/(2\log_{1.5} m) $.
\end{lemma}

\begin{proof}
First, call a token {\it unlucky} if it has not been successfully split after the $h$ steps. Otherwise, we say a token is {\it lucky} (either an augmenting path has been found or the token has successfully split.) Consider a bipartite multi-graph where the left nodes $X$ denote the unlucky tokens and the right nodes $Y$ denote the lucky tokens. Add an edge between an unlucky token $x$ and a lucky token $y$, if $y$'s split prevents $x$ to split on one of its $h$ steps. Note that in case $x$ traveled through the same node multiple times, multiple edges can be added between the same pair of nodes. 

Observe that the degree of a left node is exactly $h \cdot \lfloor (2t-d-2)/\log_{1.5} m \rfloor$. The number of edges in the bipartite graph is $h \cdot |X| \cdot \lfloor (2t-d-2)/\log_{1.5} m \rfloor$. Also, the degree of each right node is at most $d-t+1$, as the number of requests at each node is at most $d-t+1$. Therefore,
\begin{align*}
|Y| &\geq h \cdot |X| \cdot  \lfloor (2t-d-2)/\log_{1.5} m \rfloor / (d-t+1) \\
&\geq |X| \cdot  \frac{h}{d} \cdot \left\lfloor \frac{\epsilon d - 2 }{\log_{1.5} m} \right\rfloor && t \geq \lfloor (1+\epsilon)d/2 \rfloor + 1\\
&\geq |X| \cdot  \frac{h}{d} \cdot \left(\frac{\epsilon d}{2\log_{1.5} m} \right) && \mbox{$\epsilon d \geq 4\log_{1.5} m$ and when $\log_{1.5} m \geq 2$ } \\ 
&\geq |X| \cdot (8/3)\log_{1.5}m && h \geq (16/3)\log^2_{1.5} m / \epsilon  
\end{align*}

Therefore,
\begin{align*}
|S_i| \cdot L_{i} = |X| + |Y| &\geq |X|\cdot (1 + (4/3)\log_{1.5}m) \\
&\geq |F_i| \cdot \left( \frac{3}{4} L_i \right)\cdot (1 +(4/3)\log_{1.5}m) 
\geq |F_i| \cdot \left( \frac{3}{4} L_i \right)\cdot \left( \frac{8}{3} \log_{1.5} m  \right)
\end{align*}
Thus, $|F_i| \leq |S_i|/(2\log_{1.5} m) $.
\qedhere
\end{proof}

\begin{lemma}\label{lem:reducesource}

Let $A_i \subseteq S_i \setminus F_i$ denote the set of sources that successfully find an augmenting path during level $i$.  $\left| \bigcup_{i} A_i \right| \geq |S|/2$.
\end{lemma}
\begin{proof}
 Notice that $S_{i+1} = S_{i} \setminus (A_i \cup F_i)$. First we claim that $\left|\bigcup_{i} F_i\right| \leq |S|/2$. Since $|F_i| \leq |S_i| /(2\log_{1.5}m) \leq |S| / (2\log_{1.5}m)$, we have $\sum_{i} |F_i| \leq |S|/2$. Now we show that every source $x \in S \setminus \bigcup_{i} F_i$ must be in one of $A_i$. Suppose to the contrary, there exists $x \in S$ such that $x \in S_{\log_{1.5}m} \setminus (A_{\log_{1.5}m} \cup F_{\log_{1.5}m})$. Consider the edges travel by the tokens of $x$. At level $i$, since it has successfully advanced to level $i+1$, it must have traveled at least $L_{i+1}$ edges. The total number of edges traveled by the tokens of $x$ is $\sum_{i=1}^{\log_{1.5} m} L_{i+1} > (3/2)^{\log_{1.5} m} = m$. A contradiction occurs.
\end{proof}

In \Cref{alg:augment}, level $i$ takes $O(i \cdot h)$ rounds, since the length of the path spanned by the tokens from the source is $O(i\cdot h)$ and they have to communicate with the source at the end of the level. Therefore, the total number of rounds of \Cref{alg:augment} is $O(\sum_{i=1}^{\log_{1.5} m} i\cdot h) = O(h \log^2_{1.5} m) = O((\log^4 n)/\epsilon)$.

\begin{proof}[Proof of \Cref{lem:smallbalance}]
Given a $t$-balanced coloring, we can improve it to a $(t-1)$-balanced coloring by calling \Cref{alg:augment} repeatedly. By \Cref{lem:reducesource}, each invocation of \Cref{alg:augment} finds $|S|/2$ almost edge-disjoint augmenting paths from distinct sources. Of those $|S|/2$ augmenting paths, at least $1/d$ fraction will be accepted and augmented, since there can be at most $d$ paths ending with the same node. Therefore, the source reduces by a $1/(2d)$ fraction for each invocation. Since $|S| \leq n$, after $O(d \log n)$ invocations, we have obtained a $(t-1)$-balanced coloring. If we iterate $t$ from $d$ to $\lfloor (1+\epsilon)d/2 \rfloor -1$, we obtain a $\lfloor (1+\epsilon)d/2	 \rfloor$-balanced coloring. The total number of invocation of \Cref{alg:augment} is $O(d^2 \log n)$. Thus, the running time is $O((d^2 \log^5 n)/\epsilon)$.
\end{proof}

\subsection{Deterministic Undirected Degree Splitting for High-Degree Graphs}
\label{subsec:high-degree}
Suppose that the input is a graph with maximum degree $\Delta$. In this case, if we apply \Cref{lem:smallbalance} directly, it takes $\tilde{O}(\Delta^2/\epsilon)$ rounds to get a $\lfloor (1+\epsilon)\Delta/2 \rfloor$-balanced coloring. Here, we show a  method which removes the dependency on $\Delta$. 
\begin{theorem}\label{lem:largedegree}Suppose that $\Delta \geq \lceil 32 \log_{1.5} m /\epsilon^2 \rceil$, then a $\lfloor (1+\epsilon)\Delta/2 \rfloor$-balanced coloring can be obtained in $O((\log^7 n) / \epsilon^3)$ rounds. \end{theorem}
\begin{proof}
Let $\epsilon' = \epsilon/2$ and $d = \lceil 4\log_{1.5} m/\epsilon' \rceil$. Thus, we have $\Delta \geq 2d/ \epsilon'$. First, obtain $G'$ as follows: For each node $u$, split it into $\lceil \Delta/ d \rceil$ \emph{copy-nodes} and divide evenly the edges adjacent to $u$ between the copy-nodes such that each copy-node, except possibly one, has degree $d$.  Since $\epsilon' d \geq 4 \log_{1.5} m$, we can apply \Cref{lem:smallbalance} to get a get a $\lfloor (1+\epsilon')d/2 \rfloor$-balanced coloring in $G'$ in $O((d' \log^5 n)/ \epsilon')= O((\log^{7} n )/ \epsilon^3)$ rounds. Then, we merge the copy-nodes back. Then, for each node of $G$, the number of incident edges of each color is bounded by:

\begin{align*}
\left\lfloor \frac{(1+\epsilon')d}{2} \right\rfloor \cdot \left\lceil \frac{\Delta}{d} \right\rceil \leq  \frac{(1+\epsilon')d}{2} \cdot \left( \frac{\Delta}{d} + 1 \right)
	&\leq \frac{\Delta}{2} \cdot \left(1 + \epsilon' \right) + d \\
	&\leq \frac{\Delta}{2} \cdot \left(1 + \epsilon' + \frac{2d}{\Delta}\right)
	\leq \frac{\Delta}{2} \cdot \left(1 + \epsilon \right) && \mbox{$\Delta \geq 2d/ \epsilon'$ }
\end{align*}
Since the number of colors must be an integer, it is a $\lfloor (1+\epsilon)\Delta/2 \rfloor$-balanced coloring. 
\end{proof}

\begin{proof}[Proof of \Cref{thm:coloring}]

Let $\epsilon' = \epsilon / (2\log_2 \Delta)$ and $\Delta_0 = \Delta$. Suppose that $\Delta_0 \leq \lceil 32 \log_{1.5} m /\epsilon'^2\rceil$, we can use Panconesi and Rizzi's algorithm \cite{panconesi-rizzi} that runs in $O(\Delta_0 + \log^{*} n)$ rounds to get $(2\Delta_0-1)$-edge coloring. Otherwise, we apply \Cref{lem:largedegree} to get a $\lfloor (1+\epsilon')\Delta/2 \rfloor$-balanced coloring. For the subgraph consists of red edges and the subgraph consists of blue edges, we recursively apply this procedure on both of them in parallel with a new maximum degree $\Delta_{i+1} = \lfloor (1+\epsilon')\Delta_i /2 \rfloor$. Let $t$ be the level where the recursion halts. That is, $t$ is the smallest integer such that $\Delta_t \leq \lceil 32 \log_{1.5} m /\epsilon'^2\rceil$. The recursion will stop at level $t$, where we will apply Panconesi and Rizzi's algorithm to get an $(2\Delta_{t} - 1)$-edge coloring on each subgraph. Since the number of subgraphs at level $t$ is $2^{t}$, the total number of color used is 
\begin{align*}
(2\Delta_t - 1) \cdot 2^{t} &\leq 2 \cdot (1+\epsilon')^t \Delta \ \\
&\leq 2\Delta + 4t\epsilon'\Delta  && \mbox{$(1+x)^n \leq 1+2nx$ for $0 \leq nx \leq 1$} \\
&\leq (1 + 2t\epsilon') \cdot 2\Delta \\
&\leq (1 + \epsilon) \cdot 2\Delta && t \leq \log_{2} \Delta
\end{align*}

We apply the balanced coloring procedure for $O(\log \Delta)$ rounds, each takes $O(\log^{10}n /\epsilon^3)$ rounds by \Cref{lem:largedegree}. At the last level, Panconesi and Rizzi's algorithm takes $O(\Delta_k + \log^{*} n) = O(\log n / \epsilon'^2 + \log^{*} n) = O(\log^3 n/ \epsilon^2)$ rounds. Therefore, the total number rounds is: $O(\log^{11} n / \epsilon^3)$.
\end{proof}

\bibliographystyle{alpha}
\bibliography{ref}
\appendix
\section{Missing Proofs of \Cref{sect:sinkless}}
\label{app:sinkless}
\subsection{Missing Proofs of \Cref{sec:sinklessLargeDeg}}
\label{app:sinklessLargeDeg}

\begin{claim} After \Cref{alg:random}, each node is incident to at least $\Delta/2$ unmarked edges.\end{claim}
\begin{proof}
Node of type I get all their $\Delta$ edges unmarked. Each other node has at most $\Delta/2$ marked edges incident to it, by definition.
\end{proof}

\begin{claim} After \Cref{alg:random}, if the unmarked edges are oriented in a way such that all bad nodes have at least one outgoing unmarked edge, then the orientation is sinkless. \end{claim}
\begin{proof}
If a node is bad, then by the assumption it will have at least one outgoing unmarked edge. Otherwise, it must have at least one outgoing marked edge by the definition of Type III.
\end{proof}

In \Cref{sec:deterministic}, we show the deterministic algorithm can be used orient the unmarked edges if the number of unmarked edges incident to each node is between $\Delta/2$ and $\Delta$. Now it remains to bound the size of the connected components induced by the bad nodes.

Let $T_I(v)$ and $T_{III}(v)$ denote the events that $v$ is Type I or Type III, respectively. 

\begin{lemma}\label{lem:typeI}For any $v\in G$, $\Pr(T_I(v)) \leq \exp(-\Delta/12)$. \end{lemma}
\begin{proof}
Let $X$ denote the number of edges incident to $v$. We have $\E[X] = \Delta/4$. By a Chernoff Bound, $\Pr(X > \Delta / 2) \leq \exp(-\Delta /12)$.
\end{proof}

\begin{lemma}\label{lem:typeIII}For any $v\in G$, $\Pr(T_{III}(v)) \leq \exp(-\Delta/8)$. \end{lemma}
\begin{proof}
The probability that an edge is marked and oriented toward $v$ is $1/8$. Therefore, the probability that no edges are marked and oriented toward $v$ is $(1- 1/8)^{\Delta} \leq \exp(-\Delta/8)$.
\end{proof}

Let $\dist(u,v)$ denote the distance between $u$ and $v$ in $G$. If $\dist(u,v) \geq 2$, then it is clearly that the event $T_I(u)$ (or $T_{III}(u))$ and $T_I(v)$ (or $T_{III}(v)$) are independent. Let $V'$ be the set of nodes that are Type I or Type III. Define $\dist(X,v) = \min_{u \in X} \dist(u,v)$. Let $E_{2,4} = \{uv \mid 2 \leq \dist(u,v) \leq 4\}$ denote the set of edges whose endpoints have distance between 2 and 4. Let $N_{k}(u) = \{x \mid \dist(x,u) \leq k\}$ be the set of nodes within distance $k$ to $u$.

\begin{lemma}\label{lem:badcomponent} Let $C$ be a connected components induced by the bad nodes. Then, there exists $S \subseteq V' \cap C$ such that $|S| \geq |C| / \Delta^2$ and $(S, E_{2,4})$ is connected.\end{lemma}

\begin{proof}
We will construct $S$ step by step. First notice that $V' \cap C$ must be non-empty, since a Type II node must be adjacent to a Type I node, which must be in $V' \cap C$. Let $u \in V' \cap C$. Initially, Let $S =\{u\}$. Now we will show how to extend $S$ by adding one node $z \in (V'\cap C) \setminus S$ into it provided that $|S| < |C| / \Delta^2$, and $z$ is connected to some node in $S$ with an edge in $E_{2,4}$.

Suppose that $|S| < |C| / \Delta^2$, then there exists a node $w$ in $C\setminus S$ such that $\dist(S, w) = 3$, since the 2-neighborhood of $S$ can only span at most $|S|(1 + \Delta + \Delta \cdot (\Delta - 1)) \leq |S|\Delta^2$ nodes and $C$ is connected. If $w$ is Type I or Type III, then $w \in V' \cap C$ and we can add $w$ to $S$. Otherwise, it is Type II, which implies it has a neighbor $z$ of Type I. We must have $2 \leq \dist(u,z) \leq 4$. Thus, we can add $z$ to $S$.
\end{proof}

Therefore, given a connected component $C$ induced the bad nodes, we can find a tree $T$ in $(V'\cap C, E_{2,4})$ such that $|T| \geq |C| / \Delta^2$. Next we show that any sufficient large tree are not likely to occur, which implies no big bad components exist.

\begin{lemma}\label{lem:badtree}
The probability that any tree $T$ with $T \subseteq (V', E_{2,4})$ and $|T| =\Omega(\log n)$ exists is at most $1/\poly(n)$.
\end{lemma}
\begin{proof}
Let $T$ be a tree such that $T \subseteq (V', E_{2,4})$ and $|T| = t$. The probability that a node in $T$ is marked as Type I and Type III (thus in $V'$) is at most $\exp(-\Delta/12)$ by \Cref{lem:typeI} and \Cref{lem:typeIII}. Therefore, the probability that $T$ occurs is at most $\exp(-t \cdot \Delta/12)$, since the events $T_I(u)$ (or $T_{III}(u)$) are independent among the nodes $u\in T$. The total possible number of such trees is at most $4^{t} n (\Delta^4)^{t-1}$. By union bounding the possible trees, the probability that any tree in $(V', E_{2,4})$ of size $t$ occur is at most $n \cdot (4\Delta^4\cdot \exp(-\Delta / 12))^t$. For $\Delta \geq 500$, this is at most $n \cdot e^{-t}$. Thus, the probability that any tree in $(V', E_{2,4})$ of size at least $10\log n$ exists is at most $(1 / n^{9}) \cdot \sum_{i=0}^{\infty} e^{-i} = 1/\poly(n)$.
\end{proof}

\begin{corollary}\label{crl:componentSize}
The probability that any connected component induced by the bad nodes has size of $\Omega(\Delta^2 \log n)$ is at most $1/\poly(n)$.
\end{corollary}
\begin{proof}
If there exists a bad connected component $C$ with size $\Omega(\Delta^2 \log n)$, then by \Cref{lem:badcomponent}, there exists $S \subset V' \cap C$ such that $|S| = \Omega(\log n)$ and $(S, E_{2,4})$ is connected. Therefore, a tree $T$ in $(V', E_{2,4})$ occurs with $|T| = \Omega(\log n)$, which happens with probability $1/\poly(n)$ by \Cref{lem:badtree}.
\end{proof}

\subsection{Generalization to Irregular Graphs with Min-Degree $d \geq 3$, and Refinements}
\label{app:sinklessGeneralizations}
In the previous subsection, we presented an $O(\log_{\Delta}\log n)$ sinkless orientation algorithm for $\Delta$-regular graphs with $\Delta>500$. Here, we extend the result to a irregular graphs with min-degree $d \geq 3$, with round complexity becoming $O(\log_{d}\log n)$. We also show how to achieve a more refined directed degree split in almost the same running time.

First let us deal with irregular graphs with min-degree $d>500$. 

\begin{lemma}\label{lem:IrregularHigh-Degree} There is a randomized algorithm that computes a sinkless orientation of graphs of minimum degree $d\geq 500$ in $O(\log_{d}\log n)$ rounds.
\end{lemma}
\begin{proof}
We transform graph $G$ into a $d$-regular structure $H$, which is essentially a graph but allowing edges with only one endpoint, which we call half-edges. For each node $v\in G$ with degree $d'>d$, remove $v$ and instead add $\lfloor d'/d \rfloor$ copy-nodes, assign $d$ of edges of $v$ to each of these copy-nodes, and mark the remaining edges. We do not need those marked edges to be oriented outwards from $v$ (or its copy-nodes). If an edge is marked by both of its end-points, drop it. Otherwise, think of it simply as a \emph{half-edge}, having only one endpoint which may wish to have this edge outgoing. Now, the graph is transformed into a new structure where each node is incident on exactly $d$ edges or half-edges. A sinkless orientation of this structure can be compute using \Cref{lem:deterministicSinkless}.
\end{proof}
We now explain how to extend the algorithm to cases where min-degree is $d\in [3, 500]$. 
\begin{lemma}\label{lem:IrregularLow-Degree}There is a randomized algorithm that computes a sinkless orientation of graphs of minimum degree $d\in[3, 500]$ in $O(\log\log n)$ rounds.
\end{lemma}
\begin{proof} Let $c$ be a small constant such that $(d-1)^{c/2}>500$. First, we find an orientation for all edges which are in cycles of length up to $3c$, in $O(1)$ rounds. This can be done easily using the method of \Cref{sec:deterministic}. This already takes care of giving an outgoing edge to nodes which are incident on these edges. We next handle the rest of the nodes.

By means of the method of the previous paragraph, we can assume without loss of generality that the graph is $d$-regular, albeit possibly having half-edges. Note that this step cannot introduce a cycle of length less than $3c$. Now, compute a Maximal $c$-Independent Set $S$, on the graph while ignoring the half-edges. This can be done in $O(d^c+\log^* n)=O(\log^* n)$ rounds using standard algorithms. Then, cluster nodes by letting each node join the cluster of the closest node in $S$, while breaking ties arbitrarily. Since we have no cycle of length less than $3c$, each cluster is a tree of depth at least $c/2$ and with at least $(d-1)^{c/2}>500$ edges connecting to other clusters, and each two clusters are connected by at most $1$ edge. We now think of contracting each cluster into one node. Communications on this contracted graph can be simulated with a constant running time overhead as each cluster has constant diameter. Since each two clusters are connected with at most $1$ edge and as each cluster has at least $500$ edges connecting to other clusters, the graph after this contraction is a simple graph with each node incident on at least $500$ edges or half-edges. We can now orient this graph using the method of \Cref{lem:IrregularHigh-Degree} in $O(\log\log n)$ rounds. At the end, each contracted cluster has at least one outgoing edge. We can then orient the edges inside the cluster towards this outgoing edge, hence ensuring that all other nodes of the cluster also have out-degree at least $1$.
\end{proof}

\begin{lemma}\label{lem:IrregularLow-Degree}There is a randomized algorithm that, for a desirably small constant $\delta>0$ and sufficiently large constant $C$, computes an orientation of $\Delta$-regular graphs with $\Delta\geq C$ in $O(\log_{\Delta}\log n)$ rounds which guarantees a lower bound of $(1/6+\delta)\Delta$ on the in-degree and out-degree of each node, with high probability\footnote{We have not tried to optimize the constants or to extend the result to irregular graphs. We believe that both should be possible without too much more effort.}.
\end{lemma}
\begin{proof}We here simply the sketch the necessary changes for obtaining this result, but defer working out the details to the full version of this paper. In the randomized algorithm of \Cref{sec:randomSinkless}, mark each edge with probability $1/3-\delta$, for a small constant $\delta$, and redefine type I bad nodes as those with more than $\Delta/3$ marked edges. Also, redefine type III bad nodes to be those which are not type I or type II but still have less than $(1/6-\delta)\Delta$ incoming marked edges or less than $(1/6-\delta)\Delta$ outgoing marked edges. It is easy to go over the analysis of this algorithm and see that bad components will induce components of size at most $\poly\log n$, with high probability. Moreover, nodes that are not bad (and thus also not bad type III) already have at least $(1/6-\delta)\Delta$ incoming edges and at least $(1/6-\delta)\Delta$ outgoing edges. 

Now we turn to the deterministic algorithm that is to be run on these bad nodes, each of which is incident on at least $2\Delta/3$ unmarked edges. Now, replace each of these bad nodes $v$ with $\lceil\Delta/6\rceil$ copy-nodes, and assign $4$ edges of $v$ to each of its copies. Leave the remaining edges as half-edges, connected only to the other endpoint. Now we are dealing with a graph $H$ where each node is incident on $4$ edges or half-edges. We will compute a sinkless and sourceless orientation of $H$, hence ensuring that each node of $G$ has at least $(1/6-\delta)\Delta$ incoming edges and at least $(1/6-\delta)\Delta$ outgoing edges. 

First, compute a sinkless orientation of $H$ using the deterministic algorithm of \Cref{sec:deterministic}. It is easy to see that in this orientation, all nodes have at least one outgoing edge and at least one incoming edge, except for long nodes which are at maximal distance from short nodes. These long nodes then have only an outgoing edge, but the rest of their edges were oriented arbitrarily. We fix these arbitrary orientations to give these long-nodes also at least one outgoing edge, hence making the overall orientation of $H$ sinkless and sourceless. Let us call those long nodes at maximal distance from short nodes \emph{leaves}. Each leaf $v$ has one of its edges, one of those that go closer to short nodes, oriented outwards. If $v$ has any other edge to a non-leaf node, orient that edge inwards, hence giving $v$ also an incoming edge and solving its case. Each remaining leaf $v$ has $4-1=3$ edges, which are either half-edges, or they connect to non-leaf nodes. Take the graph induced by the remaining leaves and these remaining edges, orient it sinkless by repeating the deterministic orientation algorithm, and then flip all of these edges. That ensures each of these remaining leaves to also have at least one incoming edge, hence giving us the desired sinkless and sourceless orientation of $H$. 
\end{proof}

\section{Edge-Coloring via Coarse-grained Degree Splitting}
\label{app:SimpleColoring}
In this section, we explain a deterministic distributed edge-coloring algorithm based on a trivial and crude degree splitting. Interestingly, this simple approach already matches the state of art for very fast algorithms. More concretely, it provides a considerably simpler method for edge-coloring that matches the bounds of Barenboim and Elkin \cite{Barenboim:edge-coloring}. See also \cite[Chapter 8]{barenboim2013monograph}.

\paragraph{The Algorithm} The algorithm is recursive, we explain one level of recursion. Consider graph $G$, and suppose its maximum degree is $\Delta$. If $\Delta = O(x)$, compute and output a  $(2\Delta-1)$ edge coloring of $G$ in $O(x+\log^* n)$ rounds, using  the classical algorithm of Panconesi and Rizzi\cite{panconesi-rizzi}. Suppose $\Delta = \Omega(x)$. Let each node split itself into $\lceil\Delta/x\rceil$ \emph{copy-nodes}, and partition its edges between these copy-nodes such that each copy-node is incident on at most $x$ edges. Call this new graph $H$. Note that $H$ has maximum degree at most $x$. Use the algorithm of Panconesi and Rizzi\cite{panconesi-rizzi} to find a $(2x-1)$ edge coloring of $H$, in $O(x+\log^* n)$ rounds. This coloring provides a $(2x-1)$ coloring of edges of $G$ such that in each color class, there are at most $\lceil\Delta/x\rceil$ edges incident on each node. That is, we have partitioned $G$ into $2x-1$ graphs $G_1$, \dots, $G_{2x-1}$, each with per-node degree at most $\lceil\Delta/x\rceil$. Now recursively run the procedure on each of these subgraphs. 

\begin{lemma}\label{Lem:simpleColoring} The algorithm works in $O((x+\log^* n)\log \Delta/\log x)$ rounds and produces a $2^{1+\log \Delta/\log x} \Delta$ edge coloring.
\end{lemma}
\begin{proof}In each iteration, the maximum degree goes down by an $x$ factor. Thus, $\log \Delta/\log x$ recursions suffice. Each recursion level takes $O(x+\log^* n)$ rounds, which means we use $O((x+\log^* n)\log \Delta/\log x)$ rounds in total. To bound the number of colors, let us consider the summation of the maximum degrees in different subgraphs. Since at the end each subgraph will be colored with about 2 factor of its max degree colors, modulo the 2 factor, this summation is an upper bound on the number of used colors. In each iteration, we lose at most a $2$ factor in this summation, because we split a graph of maximum degree $d$ into $2x-1$ subgraphs each of maximum degree at most $d/x$. Hence, after $\log \Delta/\log x$ recursions, we use $2^{1+\log \Delta/\log x} \Delta$ colors.
\end{proof}
The next corollary shows that by setting $x$ appropriately, we can reconstruct the edge-coloring results of Barenboim and Elkin \cite{Barenboim:edge-coloring}. See also \cite[Theorem 8.14]{barenboim2013monograph}.
\begin{corollary}Consider a graph $G=(V, E)$, and let $\eps>0$ be an arbitrarily small constant.
\begin{itemize}
\item[(1)] An $O(\Delta)$-edge-coloring of $G$ can be computed in $O(\Delta^\eps+\log^* n)$ time.
\item[(2)] A $\Delta^{1+o(1)}$-edge-coloring of $G$ can be computed in $O((\log\Delta)^{1+\eps}+\log^* n\cdot \frac{\log \Delta}{\log\log \Delta})$ time.
\item[(3)] An $O(\Delta^{1+\eps})$-edge-coloring of $G$ can be computed in $O(\log^* n \cdot \log \Delta)$ time.
\end{itemize}
\end{corollary}
\begin{proof} Respectively use $x=\Delta^\eps$, $x=\log^{\eps} \Delta$, or $x=2^{1/\eps}$, in \Cref{Lem:simpleColoring}.
\end{proof}

\section{Undirected Degee Splitting}

\subsection{Randomized Undirected Degree Splitting}\label{sect:randomundirected}
In this section, we give a randomized distributed algorithm for obtaining a $\lceil (1+\epsilon)\Delta/2 \rceil$-balanced coloring where $0 <\epsilon < 1$ in $O(\Delta^2 \log^4 n / \epsilon^2)$ rounds. Note that this allows one to obtain a $\lceil (\Delta+1)/2 \rceil$-balanced coloring in $O(\poly(\log n, \Delta ))$ rounds. Also note that the change from floor in the previous sections to ceiling is necessary. Consider when $\Delta=2$ and $\epsilon = 1/2$. Obtaining a 1-balanced coloring is impossible in an odd cycle. In the previous sections, the restriction that $\epsilon \Delta = \Omega(\log n)$ avoided this problem.

\begin{theorem}\label{lem:randomizedundirected} Given a graph $G$ with maximum degree $\Delta$, a $\lceil (1+\epsilon)\Delta/2 \rceil$-balanced coloring can be obtained in $O(\Delta^2 \log^4 n / \epsilon^2)$ for $0 < \epsilon < 1$. \end{theorem}

We use the same approach with that in the previous sections. Given a $t$-balanced coloring, we show how to improve the coloring to a $(t-1)$-balanced coloring. Then we iterate $t$ from $\Delta$ to $\lceil (1+\epsilon)\Delta/2 \rceil + 1$.  The only difference is that we show how to find a set of almost edge-disjoint augmenting paths of size $\Omega(|S|)$ in $O(\poly(\log n, \epsilon^{-1}))$ rounds {\it without} the restriction that $\epsilon \Delta = \Omega(\log n)$. 

The following lemma can be used to show that the size of any maximal set of almost edge-disjoint short augmenting paths is $|S|$.

\begin{lemma}\label{lem:multipleshort}Let $\mathcal{P'}$ be a set of almost edge-disjoint augmenting paths. If $s$ is not a source of an augmenting path in $\mathcal{P'}$, then an augmenting path $P$ from $s$ of length at most $O(\log n / \epsilon)$ exists and $\mathcal{P'} \cup P$ is also a set of almost edge-disjoint augmenting paths.\end{lemma}
\begin{proof}

Grow a tree from  $s$ such that the path from $s$ to the leaves are alternating paths. In each step, each leaf grows by adding all the edges that have the same color with it. Three cases may occur when a leaf tries to add an edge that does not intersect with any paths in $\mathcal{P'}$. 

 If the endpoint of an edge is an unvisited node (i.e.~the node is not in the tree) with the opposite color, then it will be added to the tree. If the endpoint of an edge is a node with the same color or an unlabelled node, then an augmenting path is found. Otherwise, if the endpoint is a visited node, we will ignore it. 

There are two cases when a leaf $u$ tries to add an edge $uv$ that intersects with some path $Q \in \mathcal{P'}$. If $u$ comes before $v$ in $Q$, then we will not add $v$ to the tree and ignore it. If $v$ comes before $u$ in $Q$, then it must be the case that $u$ and $v$ are labeled the same color. Therefore, $uv$ is the last edge of $Q$ and an augmenting path $P$ has been found from $s$. Therefore, any augmenting path found during this process must be almost edge-disjoint from $\mathcal{P}$.

Suppose that an augmenting path has not been found and $T'$ is the tree after growing $T$ by one level. We show that $|T'| \geq (1+\epsilon)|T|$. First, note that each node $u \in T$ is not the last node of any path $Q$ in $\mathcal{P}$. Let $G'= G \setminus \mathcal{P}$ be the graph obtained by deleting all the edges in $\mathcal{P}$ from $G$. Let $\indeg_{H}(u)$ denote the number of incident edges to $u$ with the opposite color in a subgraph $H$. Let $\outdeg_H(u)$ denote the number of incident edges to $u$ with the same color in $H$. We claim that for any $u \in T$, $\outdeg_{G'}(u) - \indeg_{G'}(u) \geq 2t - 2 - d$.

Note that since $u$ is a labelled node, we have $\outdeg_{G}(u) \geq t-1$, $\indeg_{G}(u) \leq d-t+1$ and so $\outdeg_{G}(u) - \indeg_{G}(u) \geq 2t - 2 - d$. Suppose that $u$ in not the first node of any paths in $\mathcal{P}$, then both $\indeg{u}$ and $\outdeg(u)$ decreases by 1 when we delete the path from $G$. This implies $\outdeg_{G'}(u) - \indeg_{G'}(u) \geq 2t-2-d$. On the other hand, if $u$ is the first node of some path in $\mathcal{P}$, then $u$ must be a source. In this case, deleting the path decreases $\outdeg(u)$ by 1. Also, since $u$ is a source, $\outdeg_{G}(u) - \indeg_{G}(u) \geq 2t - d$  Therefore, $\outdeg_{G'}(u) - \indeg_{G'}(u) \geq \outdeg_{G}(u)- 1 - \indeg_{G}(u) \geq 2t-1-d$.

For each edge $uv$ where $u,v \in T$, it must be the case that $u$ and $v$ are colored differently. Otherwise, an augmenting path would have been found. Therefore, there must be at least $\sum_{u \in T}(\outdeg_{G'}(u) - \indeg_{G'}(u)) \geq |T|\cdot (2t - 2 - d) \geq |T| \cdot ( 2(\lceil (1+\epsilon)d/2 \rceil + 1) - 2 -d) \geq \epsilon |T| \cdot d$ edges going outside of $T$. Since the maximum degree is $d$, the number of nodes added must be at least $\epsilon |T|$. Therefore, $|T'| \geq (1+\epsilon) |T|$. After $O((\log n)/\epsilon )$ levels, the tree would grow to contain more than $n$ nodes. Therefore, an augmenting path must have been found before this happens. 
\end{proof}
\begin{corollary}\label{cor:maximalalmostdis} The size of any maximal set of almost edge-disjoint augmenting paths of length $O(\log n /\epsilon)$ is $|S|$.
 \end{corollary}
\begin{proof}If $\mathcal{P}$ is a maximal set of almost edge-disjoint augmenting paths of length $O(\log n /\epsilon)$ with cardinality less than $|S|$, then there is a source $s\in S$ that does not appear in $\mathcal{P}$. We can apply \Cref{lem:multipleshort} to add an augmenting path of length $O(\log n /\epsilon)$ to $\mathcal{P}$ without violating the maximality condition. \end{proof}

\begin{lemma} A maximal set of almost edge-disjoint augmenting paths of length at most $l$ can be found in $O(l^2 \log n)$ rounds.\end{lemma}
\begin{proof}
We construct a super-graph $\mathcal{G}$ where each node in $\mathcal{G}$ denotes an augmenting path of length at most $l$. Two nodes $P_1$ and $P_2$ are connected if they are not ordered disjoint or if they share the same source. A maximal independent set (MIS) in $\mathcal{G}$ corresponds to a maximal set of almost edge-disjoint augmenting paths of length at most $l$ in $G$.

To simulate the computation on $\mathcal{G}$, we let each source $s \in S$ to be responsible for the nodes whose corresponding augmenting path start at $s$. Then, one round of communication in $\mathcal{G}$ can be simulated in $O(l)$ rounds. The total number of nodes in $\mathcal{G}$ is $O(n^{l})$. Therefore, Luby's MIS algorithm takes $O(\log (n^l)) = O(l \log n)$ rounds in $\mathcal{G}$. Since each round in $\mathcal{G}$ can be simulated in $O(l)$ rounds, the number of rounds needed to simulate Luby's algorithm is $O(l^2 \log n)$.
\end{proof}

By setting $l = O(\log n/ \epsilon)$, we can find $|S|$ almost edge-disjoint augmenting paths in $O(\log^3 n / \epsilon^2)$ rounds. Then, at least $1/\Delta$ fraction of the paths can be augmented, since there are at most $\Delta$ augmenting paths could end at the same node and one of them will be augmented. Therefore, since the number of sources is at most $n$, after $O(\Delta \log n)$ iterations, all the sources are saturated. Also, since we iterate $t$ from $\Delta$ to $\lceil(1+\epsilon)\Delta /2 \rceil$, the total number of rounds is $O(\Delta^2 \log^4 n / \epsilon^2)$.

\subsection{Randomized Edge Coloring}\label{sec:randomizedcoloring}
In this section, we show how to obtain a faster randomized edge-coloring algorithm using $(4+\epsilon)\Delta$ colors, in $\poly(\log\log n)$ rounds. This is by combining the graph shattering technique with our deterministic algorithm. For $\Delta \in [\poly(\log\log n), \frac{\log n}{\poly(\log\log n)}]$, this is faster than the $O(\sqrt{\Delta}\log \Delta)$-round $O(\Delta)$-edge-coloring that follows from the work of Barenboim\cite{Barenboim:2015}, and the $O(\log^* \Delta \cdot \max\{1,\frac{\log n}{\Delta^{1-o(1)}}\})$-round $((1+\eps)\Delta)$-edge-coloring result of Elkin, Pettie, and Su\cite{EPS15}.

\begin{theorem}Given a graph $G$ and $0 < \epsilon < 1$, a $(4+\epsilon)\Delta$-edge coloring can be obtained in $O((\log^{11} \log n) / \epsilon^3)$ rounds. \end{theorem}
\begin{proof}
First we will assume that $\Delta = O(\log ^2 n)$. Since for $\Delta = \Omega(\log^2 n)$, by using Elkin et al.'s algorithm \cite[Theorem 2.1]{EPS15}, a $(1+o(1))\Delta$-edge-coloring can be obtained in $O(\log^{*} n)$ rounds. Let $\epsilon' = \epsilon / 4$. We will divide the $(4+\epsilon)\Delta$ colors into two sets $C_1$ and $C_2$ with an equal size, so each set consists of $2(1+\epsilon')\Delta$ colors. 

\paragraph{Pre-shattering} Let $x = \epsilon'^2 \Delta/ (18\log \Delta)$. First we partition the edges randomly to form subgraphs $G_1,G_2,\ldots G_x$. The expected degree of each node $u$ in $G_i$ is at most $\Delta / x = (18\log \Delta)/\epsilon'^2$. By Chernoff Bound, $\Pr(\deg_{G_i}(u) \geq (1+\epsilon') \Delta/ x) \leq e^{-\epsilon'^2 (18\log \Delta) / (3\epsilon'^2)} \leq (1/\Delta)^{6}$. For $v \in G$, we mark $v$ as a Type I node if there exists $1 \leq i \leq x$, such that $\deg_{G_i}(u) \geq (1+\epsilon') \Delta/ x$. By an union bound over $1 \leq i \leq x$, the probability that $v$ is Type I is at most $x \cdot (1/\Delta)^6 \leq 1/\Delta^5$.

If $u \in G$ is a not Type I node but it is adjacent to a Type I node, then it is a Type II node. Type I nodes and Type II nodes are the bad nodes. Let $B$ be the set of bad nodes and $V'$ be the set of Type I nodes. 

\paragraph{Post-shattering} First note that every subgraph $G_i[V\setminus V']$ has maximum degree bounded by $(1+\epsilon')\Delta/x$. We divide the colors in $C_1$ evenly into $C_{i1}, C_{12}, \ldots, C_{1x}$ so that each has size $2(1+\epsilon')\Delta/x$. We will run Panconesi and Rizzi's algorithm \cite{panconesi-rizzi} to get a $2(1+\epsilon')\Delta/x$-edge coloring in $O(\Delta / x + \log^{*} n) = O((\log \Delta )/\epsilon^2 + \log^{*} n) = O((\log \log n) / \epsilon^2)$ rounds on each $G_i[V\setminus V']$ in parallel with the color set $C_{1i}$.

Now the uncolored edges must be the ones that are adjacent to $V'$. Since $B = N(V') \cup V'$, all the uncolored edges must be in $G[B]$. Similar to the analysis of sinkless orientation in \Cref{sect:sinkless}, we will show that each component in $G[B]$ has their size bounded by $\polylog(n)$. Then, we will apply our determinisitic algorithm on each component with the color set $C_2$.

Let $\dist(u,v)$ denote the distance between $u$ and $v$ in $G$. If $\dist(u,v) \geq 2$, then it is clearly that the event $u$ becomes Type I and the event $v$ becomes Type I are independent.  Let $E_{2,4} = \{uv \mid 2 \leq \dist(u,v) \leq 4\}$ denote the set of edges whose endpoints have distance between 2 and 4. 

\begin{lemma}\label{lem:badcomponent2} Let $C$ be a connected components in $G[B]$. Then, there exists $S \subseteq V' \cap C$ such that $|S| \geq |C| / \Delta^2$ and $(S, E_{2,4})$ is connected.\end{lemma}
\begin{proof}The proof is exactly the same with that of \Cref{lem:badcomponent} except that now we are omitting Type III nodes. \end{proof}

\begin{lemma}\label{lem:badtree2}
The probability that any tree $T$ with $T \subseteq (V', E_{2,4})$ and $|T| =\Omega(\log n)$ exists is at most $1/\poly(n)$.
\end{lemma}
\begin{proof}
Let $T$ be a tree such that $T \subseteq (V', E_{2,4})$ and $|T| = t$. The probability that a node in $T$ is marked as Type I is at most $1/\Delta^5$. Therefore, the probability that $T$ occurs is at most $1/\Delta^{5t}$, since the events each node in $T$ becomes Type I are independent. The total possible number of such trees is at most $4^{t} n (\Delta^4)^{t-1}$. By an union bound over all possible trees, the probability that any tree in $(V', E_{2,4})$ of size $t$ occur is at most $n \cdot (4\Delta^4\cdot (1/\Delta^5) )^t$. For $\Delta \geq 4e$, this is at most $n \cdot e^{-t}$. By summing over $t \geq 10\log n$, the probability that any tree in $(V', E_{2,4})$ of size at least $10\log n$ is at least $(1 / n^{9}) \cdot \sum_{i=0}^{\infty} e^{-i} = 1/\poly(n)$.
\end{proof}

Therefore, if $C$ is a connected compnent in $G[B]$ of size $\Omega(\Delta^2 \log n)$. Then by \Cref{lem:badcomponent2}, there exists a set of vertices $S \subseteq V' \cap C$ with $|S| = \Omega(\log n)$ and $(S, E_{2,4})$ is connected. Take a spanning tree $T$ of the graph $(S, E_{2,4})$. Since $|T| = \Omega(\log n)$, by \Cref{lem:badtree2}, the probability that $|T|$ exists is at most $1/\poly(n)$. Therefore, we conclude that the probability $|C| = \Omega(\Delta^2 \log n)$ is at most $1/\poly(n)$.

Since $\Delta = O(\log^{2} n)$, each component has size at most $O(\log^{3} n)$. Now we will run our deterministic $(2+\epsilon')\Delta$-edge coloring using colors in $C_2$ on each of the component. By \Cref{thm:coloring}, the running time is $O((\log^{11} \log n)/\epsilon^3)$.\end{proof}

\section{Directed Degree Splitting}

Here, we consider the problem of orienting the edges such that the out-degree and in-degree of each node are both upper bounded by $D$, for any given $D \geq \lceil (\Delta+1)/2 \rceil$. The following lemma allows us to focus on only the out-degree side of the problem, and then extend it to both out-degree and in-degree:

\begin{lemma}\label{lem:half} Let $D \geq \lceil (\Delta+1)/2 \rceil$. Let $G$ be a graph with an arbitrary orientation. Suppose that $A$ is a distributed algorithm that orients $G$ into $G'$ such that the out-degree of each node is at most $D$ in $T$ rounds with the following property: For each $u$, $\outdeg_{G'}(u) \geq \min(\outdeg_{G}(u), D)$. Then in $O(T)$ rounds, $G$ can be oriented such that both the in-degree and the out-degree of each node is at most $D$. \end{lemma}
\begin{proof} 
First run $A$ on $G$ to obtain an orientation with out-degree at most $D$, say the resulting graph is $G_1$. Then, we reverse each edge in $G_1$ to obtain $G_2$. $G_2$ is a graph such that the in-degree of each node is at most $D$.  Now, run $A$ on $G_2$ to obtain $G_3$. If the in-degree of a node $u$ in $G_2$ is at most $D$, then 
\begin{align*}
\indeg_{G_3}(u) &= \deg(u) - \outdeg_{G_3}(u) \\
&\leq \deg(u) - \min(\outdeg_{G_2}(u),D) && \mbox{by the property of $A$} \\ 
&\leq  \deg(u) - \min(\deg(u) - D ,D) && \indeg_{G_2}(u) \leq D \\
&\leq \deg(u) - (\deg(u) - D) = D && \mbox{$D > \Delta/2$ and $\deg(u) \leq \Delta$  }
\end{align*}
Therefore, the in-degree of every node is still at most $D$ in $G_3$.
\end{proof}

By \Cref{lem:half}, it suffices to develop algorithms that orient the graph such that each node's out-degree is bounded by $D$, with the stated additional property. Our augmentation-based approach satisfies this property, because we only decrease the out-degree in nodes with out-degree at least $D+1$.

For deterministic algorithms, we will use the same approach as in \Cref{sec:undirected}. Recall that in \Cref{sec:undirected} we assumed that a $t$-balanced coloring is given, and we showed how to improve it to a $(t-1)$-balanced coloring. Similarly, here we assume that we have an orientation with the out-degree of each node upper bounded by $t$. We use the same approach to improve it to an orientation with the out-degree of each node upper bounded by $t-1$. The outer-loop will iterate $t$ from $\Delta$ to $D+1$.

The definition of an augmenting path is much more straightforward here. An augmenting path is just a directed path that starts from a node with out-degree equals to $t$ and ends at a node with outdegree at most $t-2$. We augment along a path by reversing the orientation of each edge. After we augmented along an augmenting path, the out-degree of the first node decreased by 1 and the out-degree of the last node increased by 1. All the other nodes remain to have the same out-degree. Since $t > D$, we will not decrease the out-degree of a node with out-degree at most $D$. This satisfies the property stated in \Cref{lem:half}.

Let $S$ be the source nodes, which are the nodes with out-degree $t$. Similar to \Cref{sec:undirected}, we will find a set of $\Omega(|S|)$ edge-disjoint augmenting paths from distinct sources in $S$ here. However, the singularity that appeared in \Cref{sec:undirected} where the last edges of augmenting paths in a set of almost edge-disjoint augmenting paths may overlap do not appear here. We just need to find a set of edge-disjoint augmenting path. By repeating the same arguments as in \Cref{sec:undirected}, we can obtain the analogues of \Cref{lem:smallbalance} and \Cref{lem:largedegree}, which we state without proofs in this version.

\begin{lemma}\label{lem:smallorient}Given a graph $G$ with maximum degree $d$, an orientation where the out-degree and the in-degree of each node is at most $\lfloor (1+\epsilon)d/2 \rfloor$ can be obtained in $O((d^2 \log^5 n) /\epsilon)$ rounds provided that $(4\log_{1.5}m) / d < \epsilon < 1$. \end{lemma}

\begin{theorem}\label{lem:largeorient} Suppose that $\Delta \geq \lceil 32 \log_{1.5} m /\epsilon^2 \rceil$, then an orientation where the out-degree and the in-degree of each node is at most $\lfloor (1+\epsilon)\Delta/2 \rfloor$  can be obtained in $O((\log^7 n) / \epsilon^3)$ rounds. \end{theorem}

\subsection{Randomized Directed Degree Splitting and Graphs with Bounded Arboricity}
In this section, we show how to obtain an orientation with the out-degree of each node bounded by $\lceil (1+\epsilon)a \rceil$ for any $0 < \epsilon <1$ in graphs with arboricity bounded by $a$.

\begin{theorem}\label{thm:directedDegreeSplit} There is a randomized distributed algorithm that in $O(\log^4 n/\eps^3)$ rounds, produces an orientation of $a$-arboricity graphs with per-node out-degree at most $\lceil (1+\epsilon)a \rceil$, for any $0 < \epsilon <1$.
\end{theorem}
Notice that by setting $\epsilon=1/a$, we get an orientation with per-node out-degree at most $a+1$ in $O(a^3\log^4 n)$ rounds. Moreover, notice that each graph with maximum degree $\Delta$ has arboricity $a\leq \Delta/2$. Hence, the above theorem already supplies an orientation with per-node out-degree at most $\lceil (1+\epsilon)\Delta/2 \rceil$, in $O(\log^4 n/\eps^3)$ rounds, which can again be made a much finer orientation with out-degree at most $\lceil (\Delta+1)/2 \rceil$, in $O(\Delta^3\log^4 n)$ rounds, by setting $\epsilon=1/\Delta$. Moreover, our algorithm is again augmentation-based and it satisfies the properties in \Cref{lem:half}, which can be used to satisfy the reqirement on both the in-degree and the out-degree of each node.

\begin{corollary}\label{lem:randomizeddirected} Given a graph $G$ with maximum degree $\Delta$, an orientation where the in-degree and the out-degree of each node is bounded by $\lceil (1+\epsilon)\Delta/2 \rceil$ can be obtained in $O(\log^4 n / \epsilon^3)$ for $0 < \epsilon < 1$. \end{corollary}

Note that in comparision with our randomized algorithm stated in \Cref{lem:randomizedundirected} of \Cref{sect:randomundirected}, we have removed the dependency on $\Delta$ in the running time. This is because in \Cref{sect:randomundirected} there is a singularity on the augmenting paths who intersects on the last edge. Here, the structure of the augmenting paths is  the same with the flow networks. This allows us to apply the known techniques such as the arguments of blocking-flows. 

We now explain our method for achieving \Cref{thm:directedDegreeSplit}. Let $G_0$ be a directed graph obtained by orienting the original graph arbitrarily. Let $D = \lceil(1+\epsilon)a \rceil$. Define an augmenting path to be a directed path starting from a node with out-degree at least $D+1$ and ends at a node with out-degree at most $D-1$.  By augmenting along an augmenting path, we flip all the edges to the reverse direction. After the augmentation, the out-degree of the starting node decreased by 1 and the out-degree of the ending node increased by 1.

Define $G'_0$ to be a directed graph by adding a source node $s$ and a sink node $t$ to $G_0$. Also, add $\outdeg_{G_0}(u) - D$ edges from $s$ to every node with degree at least $D+1$ and add $D - \outdeg_{G_0}(u)$ edges from every node with degree at most $D-1$ to $t$. Now we will do multiple augmentations on $G'_0$ to obtain $G'_1 \ldots G'_l$. Define $G_i$ to be $G'_{i} \setminus \{s,t\}$. First note that $\dist_{G'_0}(s,t) \geq 3$. In step $i$, we will find a maximal set of edge-disjoint paths of length $3+i$ from $s$ to $t$ in $G'_i$ and augment along them to obtain $G'_{i+1}$. The standard blocking-flow type argument shows that the distance from $s$ to $t$ increases after the augmentation. 

\begin{lemma}\label{lem:shortest} $\dist_{G'_i}(s,t) \geq 3 + i$ for $0 \leq i \leq l$. \end{lemma}
\begin{proof}
We will show by induction that $\dist_{G'_{i}}(s,t) \geq 3 + i$. When $i=0$, it is true that $\dist_{G'_{0}}(s,t) \geq 3$. Suppose that it is true that $\dist_{G'_{i}}(s,t) \geq 3+i$. Consider $G'_{i+1}$. Let $\mathcal{P}$ be a maximal set of paths of length $3+i$ from $s$ to $t$ in $G'_{i}$. If $\mathcal{P}$ is non-empty, then all paths in it must have length $3+i$ and so they are the shortest paths.

First we show that $\dist_{G'_{i+1}}(x,t) \geq \dist_{G'_{i}}(x,t)$ for every $x$ by induction on $\dist_{G'_{i}}(x,t)$. For a node $x$ with $\dist_{G'_{i}}(x,t) = 1$, it is obviously true that $\dist_{G'_{i+1}}(x,t)\geq 1$ since $x \neq t$. Suppose that it is true that $\dist_{G'_{i+1}}(x,t) \geq \dist_{G'_{i}}(x,t)$ for all $\dist_{G'_{i}}(x,t) < k$. For a node $x$ with $\dist_{G'_{i}}(x,t) = k$, suppose that there is a path $P$ of length at most $k-1$ from $x$ to $t$ in $G'_{i+1}$. Then it must be the case that $P$ intersects with some path in $\mathcal{P}$. Let $uv$ be the first edge in $P$ that has a non-empty intersection with a path (say, $Q$) in $\mathcal{P}$. Without loss of generality, assume that $\dist_{G'_{i}}(u,t) = \dist_{G'_{i}}(v,t) + 1$. We have:
\begin{align}
\dist_{Q}(s,u) + \dist_{Q}(v,t) &= 3+(i-1) \label{eqn:first} \\ 
\dist_{P}(x,v) + \dist_{P}(u,t) &= k-2  \label{eqn:second}
\end{align}
If $\dist_{P}(x,v) + \dist_{Q}(v,t) \leq k-1$, then it implies that $\dist_{G'_{i}}(x,t) \leq k-1$ and a contradiction occurs. Otherwise, $\dist_{P}(x,v) + \dist_{Q}(v,t) \geq k-2$, we have
\begin{align*}
\dist_{Q}(s,u) + \dist_{Q}(u,t) &= \dist_{Q}(s,u) + \dist_{G'_{i}}(u,t) \\
&\leq   \dist_{Q}(s,u) + \dist_{G'_{i+1}}(u,t) && \mbox{by induction hypothesis} \\
&\leq \dist_{Q}(s,u) + \dist_{P}(u,t) \\
&\leq 3+(i-1). && \mbox{by (\ref{eqn:first}) and (\ref{eqn:second})}
\end{align*}
This contradicts with that $\dist_{Q}(s,t) = 3+i$.
Therefore, we have $\dist_{G'_{i+1}}(x,t) \geq \dist_{G'_{i}}(x,t)$ for all $x$. Now suppose to the contrary that a path $P'$ from $s$ to $t$ with length at most $3+i$ exists in $G'_{i+1}$. If $P'$ does not intersect with any paths in $\mathcal{P}$, then it must have length $d$ and so it must have been included in $\mathcal{P}$. Otherwise, $P'$ must intersect with some edge of the paths in $\mathcal{P}$. Suppose that $u'v'$ is the first edge $P'$ has a non-empty intersection with and $u'v' \in Q' \in \mathcal{P}$. Without loss of generality, assume that $\dist_{G'_{i}}(u',t) = \dist_{G'_{i+1}}(v',t) + 1$. We have:
\begin{align}
\dist_{Q'}(s,u') + \dist_{Q'}(v',t) &= 3 + (i-1) \label{eqn:2first} \\ 
\dist_{P'}(s,v') + \dist_{P'}(u',t) &\leq 3 + i  \label{eqn:2second}
\end{align}
If $\dist_{P'}(s,v') + \dist_{Q'}(v',t) \leq 3 + (i-1)$, then $\dist_{G'_{i}}(s,t) \leq 3+(i-1)$, a contradiction occurs. Otherwise, we have $\dist_{P'}(s,v') + \dist_{Q'}(v',t) \geq 3 + i$ and
\begin{align*}
\dist_{Q'}(s,u') + \dist_{Q'}(u',t) &= \dist_{Q'}(s,u') + \dist_{G'_{i}}(u',t) \\
&\leq \dist_{Q'}(s,u') + \dist_{G'_{i+1}}(u',t) \\
&\leq \dist_{Q'}(s,u') + \dist_{P'}(u',t) \\
&\leq 3 + (i - 1) && \mbox{by (\ref{eqn:2first}) and (\ref{eqn:2second})}
\end{align*}
This contradicts with that $\dist_{Q'}(s,t) = 3+ i$. Therefore, we must have $\dist_{G'_{i+1}}(s,t) \geq 3 + (i+1)$
\end{proof}

\begin{lemma}\label{lem:subset} Let $S_i$ denote the set of node with out-degree at least $D+1$ in $G_i$ and $T_i$ denote the set of node with out-degree at most $D-1$ in $G_i$. We have $S_{i+1} \subseteq S_{i}$ and $T_{i+1} \subseteq T_{i}$.\end{lemma}
\begin{proof}
Let $\mathcal{P}$ be the set of shortest path from $s$ to $t$ in $G'_i$. Each $P \in \mathcal{P}$ must contain exactly one node in $S_i$ as the second node, since $P$ is a shortest path. Similarly, $P$ must contain exactly one node in $T_i$ as the second to the last node. Therefore, augmenting along $P$ can only decrease the out-degree of the node in $S_i$ by 1 and increase the out-degree of $T_i$ by 1 in $G_i$. The out-degrees of other nodes remain the same. Therefore, it is impossible to create new nodes with out-degree at least $D+1$ or new nodes with out-degree at most $D-1$.
\end{proof}

\begin{lemma}No augmenting path of length at most $l$ exists in $G_l$. \end{lemma}
\begin{proof}
By \Cref{lem:shortest}, we have $\dist_{G'_l}(s,t) \geq l + 3$. Suppose that $P$ is an augmenting path with length at most $l$ in $G_l$. Let $x$ and $y$ be the starting node and the ending node of $P$. Since the out-degree of $x$ in $G_l$ is at least $D+1$, by \Cref{lem:subset}, the out-degree of $x$ in $G_0$ is also at least $D+1$. This implies $\outdeg_{G_0}(x) - D$ edges have been added from $s$ to $x$. Since the out-degree of $x$ is the same in $G'_0, \ldots G'_l$ and the out-degree of $x$ in $G_l$ is at least $D+1$, it must be the case that there is at least one edge going from $s$ to $x$ in $G'_l$. Similarly, it must be the case that there is at least one edge going from $y$ to $t$ in $G'_l$. This implies $P$ can be extended to a direct path $sPt$ of length $l+2$ in $G'_l$, which contradicts with the fact $\dist_{G'_l}(s,t) \geq l + 3$.
\end{proof}

   Therefore, no augmenting path of length at most $l$ exists in $G_l$. However, by setting $l = \Omega((\log n)/\epsilon)$, this contradicts with the following lemma.
\begin{lemma}\label{lem:shortaugboundarbor}Let $G$ be a directed graph. Let $D = \lceil (1+\epsilon)a \rceil$, where $a$ is the arboricity. Suppose that the out-degree of $u$ is at least $D+1$, then an augmenting path from $u$ of length $O((\log n) / \epsilon)$ exists.  \end{lemma}
\begin{proof} Let $B_i$ be the set of all nodes reachable from $u$ by directed paths of length $i$. We show by induction that $|B_i| \geq (1+\eps)^{i-1}$, unless $B_i$ includes a node of out-degree less than $t$ in which case, we have found an augmenting path of length $i$. The base case $i=0$ is trivial as $B_0=\{u\}$. For the inductive step, suppose that all nodes in $B_i$ have out-degree at least $D \geq (1+\eps)a$. Then, $B_i$ is incident on at least $|B_i| \cdot D $ outgoing edges. Since all these edges have both their endpoints in $B_{i+1}$, and as $B_{i+1}$ has at most $a|B_{i+1}|$ edges by definition of arboricity, we get that $|B_{i+1}|\geq \frac{D}{a} |B_{i}| \geq (1+\eps) \cdot |B_{i}| \geq (1+\eps) \cdot (1+\eps)^{i-1}=(1+\eps)^{i}$. Now, since this growth cannot continue beyond $h=\log_{1+\eps} n = O(\frac{\log n}{\eps})$ hops, as that would exhaust the graph, we get that there must be a node of out-degree less than $D$ within $h$ hops, i.e., an augmenting path of length at most $h=O(\frac{\log n}{\eps})$.
\end{proof}

\begin{lemma}Suppose the current graph is $G_i$. Then, $G_{i+1}$ can be computed in $O(i^2 \log n)$ round.\end{lemma}
\begin{proof} Recall that $G'_{i+1}$ is  obtained by augmenting along a maximal set of edge-disjoint paths from $s$ to $t$ of length $2+i$ in $G'_{i}$.  We consider the supergraph $\mathcal{G}$, where each node denotes a path from $s$ to $t$ in $G'_i$ of length $2+i$. An edge is added between two nodes, if the corresponding paths intersects at some edges. Then, a maximal independent set (MIS) in $\mathcal{G}$ corresponds to a maximal set of edge-disjoint paths from $s$ to $t$ of length $i$ in $G'_{i}$.

We will show how to how to simulate the computation of the MIS in $\mathcal{G}$ when the underlying network is $G_i$. Let $S_i$ denote the set of nodes with out-degree at least $D+1$ in $G_i$. Each node $x \in S_i$ is responsible for the paths from $s$ to $t$ of length $2+i$ whose second node is $x$ (the first node is $s$). Therefore, each node in $\mathcal{G}$ is taken care by some node $x \in S_i$. One round of communication between an edge in $\mathcal{G}$ can be simulated in $O(i)$ rounds in $G_i$, since the length of the paths is $2+i$. Since the number of paths of length $2+i$ is at most $n^{2+i}$, Luby's MIS algorithm takes $O(\log n^{2+i}) = O(i \log n)$ rounds. It takes $O(i)$ rounds to simulate a round in $\mathcal{G}$. Therefore, the running time is $O(i^2 \log n)$.

The total running time is $\sum_{i=1}^{l} O(i^2 \log n) = O(\log^4 n / \epsilon^3)$.
\end{proof}

\paragraph{Turning Low Out-Degree Orientations to Forest Decomposition} Above, we explained a method for obtaining an orientation with out-degree $a(1+\eps)$. This is immediately a decomposition of the edges into $a(1+\eps)$ pseudo-forests. Recall that a pseudo-forest is a graph where each connected component is a pseudo-tree, that is a tree with the exception of having at most one more edge, which creates one cycle. If we let each node number its at most $a(1+\eps)$ out-going edges by numbers $1$, $2$, \dots, $a(1+\eps)$ uniquely, then the edges of each number form a pseudo-forest, as they form a graph with per-node out-degree at most $1$. For practically all the distributed applications that we are aware of, graphs with out-degree $1$ are as good as trees. However, from an aesthetic viewpoint, having a decomposition into actual forests would be much nicer. We next explain a method for decomposing into forests in high-arboricity graphs. Indeed, the decomposition will have an extra property which might be quite useful in the distributed context: 
\begin{lemma}\label{lem:forestDecomp}
There is a randomized $O(\log n)$-round algorithm that for graphs of arboricity $a=\Omega(\log n/\eps^2)$, transforms orientations with out-degree at most $a(1+\eps)$ to an edge-decomposition into $a(1+8\eps)$ forests, with high probability. Moreover, except for $O(\eps)$ fractions of the forests, each connected component in the other forests is merely a star-graph, that is, a tree with diameter $2$. 
\end{lemma}

\begin{proof}
We first randomly build $a(1+\eps)$ primary forests such that each node $v$ has at most $3\eps a$ of its outgoing edges not put in these forests. We then put these left-over edges into $7\eps a$ additional forests, hence getting a decomposition into $a(1+8\eps)$ forests. The connected components of the primary forests will be stars.

First, notice that some nodes might have out-degree less than $a(1+\eps)$. For simplicity, we first remove this imperfection, by giving each node $v$ with out-degree $d_v$ exactly $a(1+\eps)d_v$ outgoing edges that go to dummy nodes. These dummy nodes are just simulated by $v$ and no other actual node will need to interact with them. Now, each real node has out-degree $a(1+\eps)$. Once we have the forest decomposition, we will drop these edges going to the dummy nodes.

Call each real node $v$ active in each of the primary $a(1+\eps)$ forests with probability $q=\frac{1-\eps}{1+\eps}$ independently. Given that $a=\Omega(\log n/\eps^2)$, by a Chernoff bound, the number of active forests in each node is a number in $[a(1-2\eps), a]$, with high probability. 

Now, we find the outgoing edges of each node in each of its active forests. For each forest $i$ and each node $v$, we will find an outgoing edge of $v$ that goes to a neighbor $u$ who is not active in forest $i$ (or a dummy neighbor $u$). We claim that node $v$ can find a collection of such edges, one for each of its active forests, with high probability, given that $a=\Omega(\log n)$. The argument is as follows: we will find these edges for the active forests of $v$ greedily, and one by one. In each step, when looking for an edge for forest $i$, there are at least $a\eps$ outgoing edges of $v$ remaining. This is because $v$ has $a(1+\eps)$ outgoing edges and we only find edges for at most $a$ active forests. Now each of the endpoints of these remaining outgoing edges is active in the current forest $i$ with probability $q=\frac{1-\eps}{1+\eps}$. Hence, with probability at least $1-(\frac{1-\eps}{1+\eps})^{a\eps} \geq 1- 1/\poly(n)$, at least one of these outgoing neighbors is not active in forest $i$. We assign the outgoing edge to that neighbor to forest $i$. Also, notice that this process can be done in just $1$ round, by each node informing its incoming neighbors of its active forests, and then each node $v$ picking its outgoing edges for its active forests. At the end, we have created $a(1+\eps)$ forests, where indeed each connected component is a star. 

What remains for each node is at most $3\eps a$ outgoing edges. These can be put in $7\eps a$ additional forests in $O(\log n)$ rounds, by a method of Barenboim and Elkin\cite{barenboim2010sublogarithmic}. Hence, we have our desired decomposition into $a(1+8\eps)$ forests.
\end{proof}
\end{document}